\documentclass[11pt]{amsart}
\usepackage{amsxtra}
\usepackage[all]{xy}
\usepackage{amssymb}
\usepackage{amsmath}
\usepackage{tikz-cd}
\usepackage{array}
\usepackage{booktabs}
\usepackage{ifpdf}
\ifpdf
\else
\fi
\usepackage{tikz}
\usepackage[active]{srcltx}

\theoremstyle{plain}

\newcommand{\cleqn}{\setcounter{equation}{0}}
\newcommand{\clth}{\setcounter{theorem}{0}}
\newcommand {\sectionnew}[1]{\section{#1}\cleqn\clth}

\newtheorem{theorem}{Theorem}[section]
\newtheorem{lemma}[theorem]{Lemma}
\newtheorem{definition-theorem}[theorem]{Definition-Theorem}
\newtheorem{proposition}[theorem]{Proposition}
\newtheorem{corollary}[theorem]{Corollary}
\newtheorem{definition}[theorem]{Definition}
\newtheorem{example}[theorem]{Example}
\newtheorem{remark}[theorem]{Remark}
\newtheorem{notation}[theorem]{Notation}
\newtheorem{assumption}[theorem]{Assumption}
\newtheorem{lemma-definition}[theorem]{Lemma-Definition}
\newtheorem{lemma-notation}[theorem]{Lemma-Notation}
\newtheorem{question}[theorem]{Question}
\newtheorem{remark-definition}[theorem]{Remark-Definition}
\newcommand \bth[1] { \begin{theorem}\label{t#1} }
\newcommand \ble[1] { \begin{lemma}\label{l#1} }

\newcommand \bpr[1] { \begin{proposition}\label{p#1} }
\newcommand \bco[1] { \begin{corollary}\label{c#1} }
\newcommand \bde[1] { \begin{definition}\label{d#1}\rm }
\newcommand \bex[1] { \begin{example}\label{e#1}\rm }
\newcommand \bre[1] { \begin{remark}\label{r#1}\rm }

\newcommand \bnota[1] {\begin{notation}\label{n#1}\rm }
\newcommand \bas[1] { \begin{assumption}\label{a#1}\rm }

\newcommand \bqu[1] { \begin{question}\label{q#1}\rm }

\newcommand {\ele} { \end{lemma} }

\newcommand {\epr} { \end{proposition} }
\newcommand {\eco} { \end{corollary} }
\newcommand {\ede} { \end{definition} }
\newcommand {\eex} { \end{example} }
\newcommand {\ere} { \end{remark} }
\newcommand {\enota} { \end{notation} }
\newcommand {\eas} {\end{assumption}}

\newcommand {\equ} {\end{question}}

\newcommand \lb[1]{\label{#1}}

\def \a  {\mathfrak{a}}

\def \gl {\mathfrak{gl}}

\DeclareMathOperator \ad { {\mathrm{ad}} }

\newcommand{\beqa}{\begin{eqnarray*}}                     
\newcommand{\eeqa}{\end{eqnarray*}}

\def \sC {{\scriptscriptstyle C}}

\def \wF_mn {\wF_m \times \wF_n}
\def \wF_mnC {\wF_{m, n, \, \sC}}

\def \wF {\widetilde{F}}

\begin{document}

\setlength{\baselineskip}{1.2\baselineskip}
\title[From $\mathcal{O}$-operators of 3-Dimensional 3-Lie algebras to 3-Pre-Lie algebras]
{From $\mathcal{O}$-operators of  3-Dimensional 3-Lie algebras to low-dimensional 3-Pre-Lie algebras}
\author{ Ziying Cheng}
\address{
School of Mathematical Sciences \\
Zhejiang Normal University \\
Jinhua 321004              \\
China}
\email{chengyy@zjnu.edu.cn}

\author{Chuangchuang Kang}
\address{
School of Mathematical Sciences \\
Zhejiang Normal University \\
Jinhua 321004              \\
China}

\email{kangcc@zjnu.edu.cn}
\author{Jiafeng L\"u}
\address{
School of Mathematical Sciences    \\
Zhejiang Normal University\\
Jinhua 321004              \\
China}
\email{jiafenglv@zjnu.edu.cn}
\date{}
\begin{abstract}

In this paper, we explicitly determine all $\mathcal{O}$-operators with respect to the adjoint representation of 3-dimensional complex 3-Lie algebras. Furthermore, we provide
the induced 3-Pre-Lie algebra structures and the corresponding solutions of the 3-Lie classical Yang-Baxter equation in the 6-dimensional 3-Lie algebras  $A\ltimes_{\ad^*} A^*$.
\end{abstract}

\subjclass[2010]{16T25, 17A30, 17A42, 81T30}

\keywords{$\mathcal{O}$-operators, 3-Pre-Lie algebras, 3-Lie classical Yang-Baxter equation, 3-Lie algebras}

\maketitle
\allowdisplaybreaks
\sectionnew{Introduction}\lb{intro}

In 1985, Filippov \cite{Filippov} introduced the definition of $n$-Lie algebras (also called Filippov algebras). His paper considered $n$-ary multi-linear and skew-symmetric operation that satisfy the generalized Jacobi identity, which appear in many fields in mathematics and mathematical physics. In particular, $3$-Lie algebras play an important role in string theory \cite{Bagger,Gustavsson,Matsuo,Matsuo2}, and represent the algebraic structure associated with Nambu mechanics \cite{Nambu,Takhtajan}.  See the review article \cite{$n$-ary} for more details.

In \cite{Bai-1}, the authors introduced Rota-Baxter 3-Lie algebras with weight $\lambda$ and showed that they can be derived from Rota-Baxter Lie algebras and Pre-Lie algebras and from Rota-Baxter commutative associative algebras with derivations. It is closely related to $\mathcal{O}$-operators associated to the representation of a 3-Lie algebra \cite{Bai-2}, which are also known as relative Rota-Baxter operators with weight zero \cite{Chen}. In particular, an $\mathcal{O}$-operator  with respect to the adjoint representation of a 3-Lie algebra is exactly a Rota-Baxter 3-Lie algebra with weight zero. It is connected to the 3-Lie classical Yang-Baxter equation (3-Lie CYBE), 3-Pre-Lie algebra, 3-Post-Lie algebras and $3$-Lie-dendriform algebra \cite{Bai-2,Hou,Chtioui}. For example, a 3-Pre-Lie algebra and the solution of 3-Lie CYBE  can be obtained through the $\mathcal{O}$-operator with respect to  a representation of a 3-Lie algebra. Furthermore, based on the differential graded Lie algebra that controls deformations of an $n$-Lie algebra with a representation, the authors study the deformation and cohomology theory of relative Rota-Baxter operators  on 3-Lie algebras with weight $\lambda=0$ and $\lambda\neq 0$ \cite{Tang,Guo-Qin}.

The  concept of a 3-Pre-Lie algebra was introduced in \cite{Pei-B-G}. It is a natural construction that involves the splitting of operads of Lie algebra and applying them to the 3-Lie algebra,
which  has similar fundamental properties with Pre-Lie algebras \cite{Bai-3}. For example,
a 3-Pre-Lie algebra  gives a 3-Lie algebra and its left multiplication operator gives rise to a representation of this 3-Lie algebra.  Conversely,  the $\mathcal{O}$-operator with respect to  a representation of a 3-Lie algebra gives rise to a 3-Pre-Lie algebra. Therefore, it is worthwhile to find new 3-Pre-Lie algebras. See the \cite{Bai-2} for more details.

 There
are three ways to construct 3-Pre-Lie algebra structures:  utilizing the Pre-Lie algebras equipped with a generalized trace function \cite{Chtioui},  computing the structure constants \cite{Du}, and applying $\mathcal{O}$-operators to generate new instances of  3-Pre-Lie algebras \cite{Kang}.
In this paper, based on the classification results of 3-dimensional complex 3-Lie algebras in \cite{Bai-Song-Zhang}, our goal is to provide all $\mathcal{O}$-operators on 3-dimensional complex 3-Lie algebras. From these operators, we obtain examples of 3-Pre-Lie algebras and present the corresponding solutions of the 3-Lie CYBE.

The paper is organized as follows. In Section 2, we give some elementary facts on 3-Lie
algebras,  $\mathcal{O}$-operators, $3$-Pre-Lie algebras and 3-Lie CYBE. In Section 3, we provide the sufficient and necessary condition for the $\mathcal{O}$-operator on a 3-Lie algebra, and then we give the classification theorem (Theorem \ref{O-operqtor}) of $\mathcal{O}$-operator on 3-dimensional complex 3-Lie algebra. In
Section 4, we  give the sufficient and necessary condition for $3$-Pre-Lie algebra.
We  prove that the two-dimensional 3-Pre-Lie algebras is trivial (Thorem \ref{1111}), and get thirty-one examples of 3-dimensional 3-Pre-Lie algebra by the $\mathcal{O}$-operators (Theorem \ref{3-pre-Lie example}). In Section 5, we obtain thirty-one  skew-symmetric solutions  of 3-Lie  CYBE in  $A\ltimes_{\ad^*} A^*$ (Theorem \ref{thm:3-CYBE}). Section 6 is the proof of Theorem \ref{O-operqtor}.

Throughout this paper,  all algebras are
of finite dimension and over the complex field $\mathbb{C}$, unless otherwise stated.

\sectionnew{Preliminaries}\lb{intro}
In this section, we recall notions and results on 3-Lie algebras,  $\mathcal{O}$-operators and $3$-Pre-Lie algebras which will be needed later in the paper.

\begin{definition} A \textbf{3-Lie algebra} is a pair $(A,[\cdot,\cdot,\cdot])$, where $A$ is a vector space and $[\cdot ,\cdot,\cdot]$ : $\otimes^3A \rightarrow A$ is a skew-symmetric linear map such that the following \textbf{Fundamental Identity} holds:
\begin{equation}
 \left[x_1,x_2,[x_3,x_4,x_5]\right] = [[x_1,x_2,x_3],x_4,x_5]+[x_3,[x_1,x_2,x_4],x_5]+[x_3,x_4,[x_1,x_2,x_5]],
 \end{equation}
 where $x_i\in A,~1 \leq i \leq 5$.
\end{definition}
\begin{proposition}\label{prop:3d-3Lie}{\rm(\cite{Bai-Song-Zhang})}
There is a unique non-trivial 3-dimensional complex 3-Lie algebra. It has a basis ${e_1,e_2,e_3}$ and the non-zero product is given by
\begin{eqnarray}\label{3-3-Lie algebra}
  [e_1,e_2,e_3] = e_1.
\end{eqnarray}

\end{proposition}
\begin{definition}{\rm(\cite{Kasymov})} Let $V$ be a vector space. A \textbf{representation} of a 3-Lie algebra $(A,[\cdot,\cdot,\cdot])$ on $V$ is a skew-symmetric linear map $\rho:\otimes^2\rightarrow \gl(V)$ such that for any $x_1, x_2, x_3, x_4, \in A$,
 \begin{eqnarray}
    &&\rho(x_1,x_2)\rho(x_3,x_4)-\rho(x_3,x_4)\rho(x_1,x_2)=\rho([x_1,x_2,x_3],x_4)-\rho([x_1,x_2.x_4],x_3), \\
  && \rho([x_1,x_2,x_3],x_4)=\rho(x_1,x_2)\rho(x_3,x_4)+\rho(x_2,x_3)\rho(x_1,x_4)+\rho(x_3,x_1)\rho(x_2,x_4).
 \end{eqnarray}
\end{definition}
\begin{definition}
Let $(V,\rho)$  be a representation of a 3-Lie algebra $(A,[\cdot,\cdot,\cdot])$. Define $\rho^*:\otimes^2A\rightarrow \gl(V^*)$ by
\begin{eqnarray*}
  \langle\rho^*(x_1,x_2)\alpha,v\rangle=-\langle\alpha,\rho(x_1,x_2)v\rangle, \quad \forall~  \alpha \in V^*, x_1,x_2\in A, v \in V.
\end{eqnarray*}
Then we  call $(V^*,\rho^*)$   the {\bf dual representation} of $(V,\rho)$.
\end{definition}
\begin{definition}
Let $(A,[\cdot,\cdot,\cdot])$ be a 3-Lie algebra, if the linear map $\ad:\otimes^2A \rightarrow \gl(A)$ satisfies for any $x_1, x_2, x \in A$,
\begin{eqnarray}
  \ad_{x_1,x_2}:A\rightarrow A ,\quad \ad_{x_1,x_2}x=[x_1,x_2,x],
\end{eqnarray}
then we call $(A,\ad)$ the {\bf adjoint representation} of $(A,[\cdot,\cdot,\cdot])$. The dual representation $(A^*,\ad^*)$ of  the adjoint representation $(A,\ad)$ is called the {\bf coadjoint representation}.
\end{definition}

\begin{definition} {\rm(\cite{Bai-2})}
Let A be a vector space with a linear map $ \{ \cdot ,\cdot,\cdot\}: A \otimes A \otimes A \rightarrow A$. The pair $(A, \{ \cdot ,\cdot,\cdot\} )$ is called a 3-\textbf {Pre-Lie algebra} if the following identities hold:
\begin{eqnarray}
  \{x,y,z\} &=& -\{y,x,z\}, \label{eq:skew-sym} \\
  \{x_1,x_2,\{x_3,x_4,x_5\}\} &=& \{[x_1,x_2,x_3]^c,x_4,x_5\}+\{x_3,[x_1,x_2,x_4]^c,x_5\}\label{eq:multi-11} \\
  \nonumber&&+ \{x_3,x_4,\{x_1,x_2,x_5\}\}, \\
  \{[x_1,x_2,x_3]^c,x_4,x_5\} &=& \{x_1,x_2,\{x_3,x_4,x_5\}\}+\{x_2,x_3,\{x_1,x_4,x_5\}\}\label{eq:multi-12}\\
  \nonumber&&+\{x_3,x_1,\{x_2,x_4,x_5\}\},
\end{eqnarray}
where $x,y,z,x_i\in A,~1 \leq i \leq 5$ and $[\cdot,\cdot,\cdot]^c $ is defined by
\begin{eqnarray}\label{eq:ad-3Lie}
  [x,y,z]^c=\{x,y,z\}+\{y,z,x\}+\{z,x,y\},\quad\forall~ x,y,z \in A.
\end{eqnarray}
\end{definition}
Moreover, let $(A, \{ \cdot ,\cdot,\cdot\} )$ be a  3-Pre-Lie algebra, then \eqref{eq:ad-3Lie} defines a  3-Lie algebra, which is called the sub-adjacent 3-Lie algebra of $(A, \{ \cdot ,\cdot,\cdot\} )$.

\begin{definition} {\rm(\cite{Bai-2})}
Let $(A, [\cdot,\cdot,\cdot])$ be a  3-Lie algebra and $(V,\rho)$ a representation. A linear operator $T:V\rightarrow$ A is called an $\mathcal{O}$-operator associated to $(V,\rho)$ if for all $ x_1, x_2, x_3 \in V$,  $T$ satisfies
\begin{equation}
   [Tx_1,Tx_2,Tx_3]=T\left(\rho(Tx_1,Tx_2)x_3+\rho(Tx_2,Tx_3)x_1+\rho(Tx_3,Tx_1)x_2\right).\label{eq:O-operator1.0}
\end{equation}
\end{definition}
Especially, for all $x_1, x_2, x_3 \in A$,   if $\rho$ is  the adjoint representation of $A$, then \eqref {eq:O-operator1.0} can be rewritten to
\begin{equation}
   [Tx_1,Tx_2,Tx_3]=T\left([Tx_1,Tx_2,x_3]+[Tx_2,Tx_3,x_1]+[Tx_3,Tx_1,x_2]\right).\label{eq:O-operator1.11}
\end{equation}
\begin{proposition}\label{O-3preLie}{\rm(\cite{Bai-2})}
Let $(A, [\cdot,\cdot,\cdot])$ be a  3-Lie algebra and $(V,\rho)$ a representation. Suppose that the linear map  $T:V\rightarrow$ A is  an $\mathcal{O}$-operator  associated to $(V,\rho)$. Then there exists a 3-Pre-Lie algebra structure on V given by
\begin{eqnarray} \label{eq:O-operator1.2}
   \{u,v,w\}=\rho(Tu,Tv)w, \quad\forall~u,v,w \in V.
\end{eqnarray}
\end{proposition}
\begin{remark}\label{rmk:O-and-3-pre}
Let $(A, [\cdot,\cdot,\cdot])$ be a  3-Lie algebra and T the $\mathcal{O}$-operator associated with the adjoint representation $(A, \ad)$.
 Then there exists a 3-Pre-Lie algebra structure on A given by
\begin{eqnarray} \label{eq:O-operator1.3}
   \{x_1,x_2,x_3\}=[Tx_1,Tx_2,x_3], \quad\forall~x_1,x_2,x_3 \in A.
\end{eqnarray}
\end{remark}

\begin{definition} {\rm(\cite{Bai-2})}
Let $(A,[\cdot,\cdot,\cdot])$ be  a 3-Lie algebra, and $r=\sum_{i}x_i\otimes y_i~\in A\otimes A$. Then $r$ is called a solution of the \textbf{3-Lie classical Yang-Baxter equation (3-Lie CYBE)} in  $(A,[\cdot,\cdot,\cdot])$ if
\begin{equation}\label{1}
  [[r,r,r]]=0,
\end{equation}
where\begin{eqnarray}
        [[r,r,r]]&=&\sum_{i,j,k}([x_i,x_j,x_k]\otimes y_i \otimes y_j \otimes y_k+x_i \otimes[y_i,x_j,x_k]\otimes y_j\otimes y_k \\
\nonumber&+&x_i\otimes x_j\otimes[y_i,y_j,x_k]\otimes y_k+x_i\otimes x_j \otimes x_k\otimes[y_i,y_j,y_k]).
     \end{eqnarray}

\end{definition}

\section{ $\mathcal{O}$-operators on 3-dimensional complex 3-Lie algebra}

In this section, we will provide the sufficient and necessary condition for the $\mathcal{O}$-operator on a 3-Lie algebra based on its structure constants.
Furthermore, we give the classification theorem of $\mathcal{O}$-operator on 3-dimensional complex 3-Lie algebra.

Suppose $\{e_1,e_2, \ldots,e_n\}$ is a basis of 3-Lie algebra $(A,[ \cdot ,\cdot,\cdot] )$. Choose three elements $e_i$, $e_j$, and $e_k$  from the basis, such that $ 1\leq i<j < k \leq n $ .  Then by \eqref{eq:O-operator1.11}, we have
\begin{equation}
   [Te_i,Te_j,Te_k]=T\left([Te_i,Te_j,e_k]+[Te_k,Te_i,e_j]+[Te_j,Te_k,e_i]\right). \label{eq:O-operator1.1}
\end{equation}

\begin{lemma}\label{lem:O}
Let $(A,[ \cdot ,\cdot,\cdot] )$ be a 3-Lie algebra and let $\{e_1,e_2, \ldots,e_n\}$ be a basis of $A$. For all positive integers $1 \leq i,j,k \leq n$ and structural constants $C_{ijk}^t \in \mathbb{C} $, set
\begin{equation}
[e_i,e_j,e_k]=\sum_{t=1}^nC_{ijk}^t e_t.
 \end{equation}
If $T:A\rightarrow A$ is a linear map defined by
\begin{eqnarray}
  T(e_i)=\sum_{m=1}^na_{im}e_m,\quad a_{im}\in \mathbb{C},
\end{eqnarray}
then $T$ is an  $\mathcal{O}$-operator if and only if the following equations hold:
\begin{equation}\label{eq:stru-cons2}
\sum_{m,s,v,t=1}^n\left(a_{js}\left(a_{it}a_{kv}C_{tsv}^m-a_{iv}a_{tm}C_{vsk}^t-a_{kv}a_{tm}C_{svi}^t\right)-a_{kv}a_{is}a_{tm}C_{vsj}^t\right)=0.
\end{equation}
\end{lemma}
\begin{proof}Let
$$T(e_i)=\sum_{m=1}^na_{im}e_m,~T(e_j)=\sum_{s=1}^na_{js}e_s,~ T(e_k)=\sum_{v=1}^na_{kv}e_v, \quad a_{im},a_{js},a_{kv} \in \mathbb{C}.$$
 Since $T$ is an $\mathcal{O}$-operator, the left-part of \eqref{eq:O-operator1.1} can be reduced to
\begin{eqnarray*}
 &&\left[Te_i,Te_j,Te_k\right] =\left[\sum_{m=1}^na_{im}e_m,\sum_{s=1}^na_{js}e_s, \sum_{v=1}^na_{kv}e_v\right]\\
&=&\sum_{m,s,v=1}^na_{im}a_{js}a_{kv}\left[e_m,e_s,e_v \right]
=\sum_{m,s,v=1}^na_{im}a_{js}a_{kv}\sum_tC_{msv}^t e_t\\
 &=&\sum_{m,s,v,t=1}^na_{im}a_{js}a_{kv}C_{msv}^te_t=\sum_{m,s,v,t=1}^na_{it}a_{js}a_{kv}C_{tsv}^me_m.
\end{eqnarray*}
The right-part of \eqref{eq:O-operator1.1} can be reduced to
\begin{eqnarray*}
&&  T\left([Te_i,Te_j,e_k]+[Te_k,Te_i,e_j]+[Te_j,Te_k,e_i]\right) \\
\nonumber&=& T\left(\left[\sum_{m=1}^na_{im}e_m,\sum_{s=1}^na_{js}e_s,e_k\right]
+\left[\sum_{v=1}^na_{kv}e_v,\sum_{m=1}^na_{im}e_m,e_j\right]
+\left[\sum_{s=1}^na_{js}e_s,\sum_{v=1}^na_{kv}e_v,e_i]\right]\right)\\
\nonumber&=&T\left(\sum_{m,s=1}^na_{im}a_{js}[e_m,e_s,e_k]+\sum_{v,m=1}^na_{kv}a_{im}[e_v,e_m,e_j]+\sum_{s,v=1}^na_{js}a_{kv}[e_s,e_v,e_i] \right)\\
\nonumber&=&\sum_{m,s=1}^na_{im}a_{js}T\left([e_m,e_s,e_k]\right)+\sum_{v,m=1}^na_{kv}a_{im}T\left([e_v,e_m,e_j]\right)+\sum_{s,v=1}^na_{js}a_{kv}T\left([e_s,e_v,e_i]\right)\\
\nonumber&=&\sum_{m,s=1}^na_{im}a_{js}T\left(\sum_t^nC_{msk}^t e_t\right)+\sum_{v,m=1}^na_{kv}a_{im}T\left(\sum_t^nC_{vmj}^t e_t\right)+\sum_{s,v=1}^na_{js}a_{kv}T\left(\sum_t^nC_{svi}^t e_t\right)\\
\nonumber&=&\sum_{m,s,t=1}^na_{im}a_{js}C_{msk}^t T\left(e_t\right)+\sum_{v,m,t=1}^na_{kv}a_{im}C_{vmj}^tT\left( e_t\right)+\sum_{s,v,t=1}^na_{sj}a_{vk}C_{svi}^tT\left(e_t\right)\\
\nonumber&=&\sum_{m,s,v,t=1}^na_{im}a_{js}C_{msk}^ta_{tv}e_v+\sum_{m,s,v,t=1}^na_{kv}a_{im}C_{vmj}^ta_{ts}e_s+\sum_{m,s,v,t=1}^na_{js}a_{kv}C_{svi}^ta_{tm}e_m\\
\nonumber&=&\sum_{m,s,v,t=1}^n\left(a_{im}a_{js}a_{tv}C_{msk}^te_v+a_{kv}a_{im}a_{ts}C_{vmj}^te_s+a_{js}a_{kv}a_{tm}C_{svi}^te_m\right)\\
\nonumber&=&\sum_{m,s,v,t=1}^n\left(a_{iv}a_{js}a_{tm}C_{vsk}^te_m+a_{kv}a_{is}a_{tm}C_{vsj}^te_m+a_{js}a_{kv}a_{tm}C_{svi}^te_m\right)\\
\nonumber&=&\sum_{m,s,v,t=1}^n\left(a_{iv}a_{js}a_{tm}C_{vsk}^t+a_{kv}a_{is}a_{tm}C_{vsj}^t+a_{js}a_{kv}a_{tm}C_{svi}^t\right)e_m.
\end{eqnarray*}
The identity \eqref{eq:O-operator1.1} holds if and only if
\begin{eqnarray}
 \sum_{m,s,v,t=1}^n\left(a_{it}a_{js}a_{kv}C_{tsv}^m-a_{iv}a_{js}a_{tm}C_{vsk}^t-a_{kv}a_{is}a_{tm}C_{vsj}^t-a_{js}a_{kv}a_{tm}C_{svi}^t\right)e_m=0.
\end{eqnarray}
Thus we have
\begin{eqnarray}
 \sum_{m,s,v,t=1}^n\left(a_{js}\left(a_{it}a_{kv}C_{tsv}^m-a_{iv}a_{tm}C_{vsk}^t-a_{kv}a_{tm}C_{svi}^t\right)-a_{kv}a_{is}a_{tm}C_{vsj}^t\right)=0.
\end{eqnarray}
Then we can draw the conclusion.
\end{proof}

Let $(A,[ \cdot ,\cdot,\cdot] )$ be the 3-dimensional 3-Lie algebra given by Proposition \ref{prop:3d-3Lie}.
A linear operator $T:A\rightarrow A$ is determined by
\begin{eqnarray}\label{eq:T-opre}
                \left(\begin{array}{c}T(e_1)\\T(e_2)\\T(e_3)
 \end{array}\right) = \left(\begin{array}{ccc}a_{11}&a_{12}&a_{13}\\
 a_{21}&a_{22}&a_{23}\\a_{31}&a_{32}&a_{33}
 \end{array}\right)
\left(\begin{array}{c}e_1\\e_2\\e_3 \end{array}\right),
              \end{eqnarray}
where $a_{ij}\in \mathbb{C}, 1\leq i,j\leq3$. If $T$ is an $\mathcal{O}$-operator,  then the matrix $(a_{ij})_{3\times3}$  satisfies
\eqref{eq:stru-cons2} for $ e_i, e_j, e_k \in \{e_1, e_2, e_3\}$.
From now on,  we use $O$ to denote  the matrix $(a_{ij})_{3\times3}$.
\begin{theorem}\label{O-operqtor}
Let $T:A\rightarrow A$ be an $\mathcal{O}$-operator on the
3-dimensional 3-Lie algebra defined by \eqref{eq:T-opre}.
Then the matrices of non-zero $\mathcal{O}$-operators on $A$ are given by:\\
\vspace{0.5cm}
$O_1=\left(\begin{array}{ccc} 0&0&0\\
a_{21}&a_{22}&a_{23}\\a_{31}&a_{32}&a_{33}
\end{array}\right)$;\quad
$O_2=\left(\begin{array}{ccc} a_{11}&0&0\\
 a_{21}&-a_{33}&a_{23}\\a_{31}&a_{32}&a_{33}
\end{array}\right) $;\quad
$O_3=\left(\begin{array}{ccc}0&0&a_{13}\\
a_{21}&0&a_{23}\\0&0&a_{33}
\end{array}\right)$;\\
\vspace{0.5cm}
$O_4=\left(\begin{array}{ccc}0&0&a_{13}\\
0&0&a_{23}\\-\dfrac{a_{23}}{a_{13}}&1&a_{33}
\end{array}\right)$;\quad
$O_5=\left(\begin{array}{ccc}0&0&a_{13}\\
0&a_{22}&a_{23}\\0&a_{32}&\dfrac{a_{23}a_{32}}{a_{22}}
\end{array}\right)$;\quad
$O_6=\left(\begin{array}{ccc}0&0&a_{13}\\
a_{21}&a_{22}&a_{23}\\0&0&0
\end{array}\right)$;\\
\vspace{0.5cm}
$O_7= \left(\begin{array}{ccc}0&0&a_{13}\\
a_{21}&1&a_{23}\\1&\dfrac{1}{a_{21}}&\dfrac{a_{23}}{a_{21}}+a_{13}
\end{array}\right)$;\quad
$O_8=\left(\begin{array}{ccc}1&0&1\\
a_{21}&a_{22}&a_{23}\\-1&0&-1
\end{array}\right)$;\quad
$O_9=\left(\begin{array}{ccc}1&0&1\\
0&a_{22}&0\\-1&a_{32}&-1
\end{array}\right)$;\\
\vspace{0.5cm}
$O_{10}=\left(\begin{array}{ccc}0&a_{12}&0\\
0&a_{22}&0\\a_{31}&a_{32}&0
\end{array}\right)$;\quad
$O_{11}=\left(\begin{array}{ccc}0&a_{12}&0\\
0&0&0\\a_{31}&a_{32}&a_{33}
\end{array}\right)$;
$O_{12}=\left(\begin{array}{ccc}0&a_{12}&0\\
0&a_{22}&1\\0&a_{22}a_{33}&a_{33}
\end{array}\right)$;\\
\vspace{0.5cm}
$O_{13}=\left(\begin{array}{ccc}0&a_{12}&0\\
1&a_{22}&1\\a_{31}&a_{22}a_{31}-a_{12}&a_{31}
\end{array}\right)$;\quad
$O_{14}=\left(\begin{array}{ccc}0&1&1\\
a_{21}&a_{22}&a_{22}\\-a_{21}&a_{32}&a_{32}
\end{array}\right)$;\\
\vspace{0.5cm}
$O_{15}=\left(\begin{array}{ccc}0&1&1\\
a_{32}(a_{22}-a_{23})&a_{22}&a_{23}\\0&a_{32}&a_{32}
\end{array}\right)$;
$O_{16}=\left(\begin{array}{ccc}0&1&1\\
0&a_{22}&a_{22}\\a_{22}(a_{33}-a_{32})&a_{32}&a_{33}
\end{array}\right)$;\\
\vspace{0.5cm}
$O_{17}=\left(\begin{array}{ccc}0&a_{12}&a_{13}\\
0&0&a_{23}\\0&0&a_{33}
\end{array}\right)$;
$O_{18}=\left(\begin{array}{ccc}0&1&1\\
0&a_{22}&a_{23}\\0&a_{32}&\dfrac{a_{23}a_{32}}{a_{22}}
\end{array}\right)$;
$O_{19}=\left(\begin{array}{ccc}0&1&1\\
a_{21}&1&2\\a_{21}&1&2
\end{array}\right)$;\\
\vspace{0.5cm}
$O_{20}=\left(\begin{array}{ccc}a_{11}&a_{12}&0\\
0&0&0\\a_{31}&a_{32}&0
\end{array}\right)$;
$O_{21}=\left(\begin{array}{ccc}a_{11}&a_{12}&0\\
0&-a_{33}&a_{23}\\0&-\dfrac{a_{33}^2}{a_{23}}&a_{33}
\end{array}\right)$;
$O_{22}=\left(\begin{array}{ccc}a_{33}&-a_{33}^2&0\\
1&-a_{33}&0\\a_{31}&a_{32}&a_{33}
\end{array}\right)$;\\
\vspace{0.5cm}
$O_{23}=\left(\begin{array}{ccc}a_{11}&1&0\\
1&-1&1\\ 1-a_{11}&-2&1
\end{array}\right)$;
$O_{24}=\left(\begin{array}{ccc}1&1&0\\
-1&-1&0\\a_{31}&a_{32}&-1
\end{array}\right)$;
$O_{25}=\left(\begin{array}{ccc}1&1&0\\
1&3+a_{23}a_{31}&a_{23}\\a_{31}&-\dfrac{2}{a_{23}}&-1
\end{array}\right)$;
$O_{26}=\left(\begin{array}{ccc}a_{11}&1&1\\
0&-a_{33}&a_{23}\\0&-\dfrac{a_{33}^2}{a_{23}}&a_{33}
\end{array}\right)$;
$O_{27}=\left(\begin{array}{ccc}1&1&1\\
0&1&0\\-1&a_{32}&-1
\end{array}\right)$;
$O_{28}=\left(\begin{array}{ccc}-1&1&1\\
0&a_{22}&1-a_{33}\\-1
&1-a_{22}&a_{33}
\end{array}\right)$,$a_{33}\neq1$;\\
\vspace{0.5cm}
$O_{29}=\left(\begin{array}{ccc}0&1&1\\
a_{33}(a_{23}-a_{33})&-a_{33}&a_{23}\\
0&-a_{33}&a_{33}
\end{array}\right)$;\quad
$O_{30}=\left(\begin{array}{ccc}1&1&1\\
1&0&\sqrt{3}-1\\
1&-1-\sqrt{3}&0
\end{array}\right)$;\\
\vspace{0.5cm}
$O_{31}=\left(\begin{array}{ccc}1&1&1\\
1&0&-1-\sqrt{3}\\
1&\sqrt{3}-1&0
\end{array}\right)$.
\end{theorem}
The proof of this theorem will be presented in Section 6.

\sectionnew{3-Pre-Lie algebras induced by $\mathcal{O}$-operators }\lb{intro}
In this section,  we will introduce the  structural constant of  $3$-Pre-Lie algebra. Then we  give the sufficient and necessary condition for $3$-Pre-Lie algebra  by its definition.
 In one direction,  by  the  equations of the structure constant, we  prove that the two-dimensional 3-Pre-Lie algebras is trivial.
In the other  direction,
 we  get thirty-one examples of 3-dimensional 3-Pre-Lie algebra by $\mathcal{O}$-operators given in Theorem \ref{O-operqtor}.

Let $\{e_1,e_2, \ldots,e_n\}$ be a basis of $(A,\{ \cdot ,\cdot,\cdot\} )$. For all $1 \leq i,j,k \leq n$, $ i,j,k \in N^+$, set
\begin{equation}\label{eq:stru-cons}
  \{e_i,e_j,e_k\}=\sum_t^nC_{ijk}^t e_t,
\end{equation}
where   $C_{ijk}^t\in \mathbb{C}$. And $C_{ijk}^t$ are called structure constants of  $3$-Pre-Lie algebra which satisfy
\begin{equation}
C_{ijk}^t=-C_{jik}^t.
\end{equation}
\begin{proposition}\label{4.1}
If $\{ \cdot ,\cdot,\cdot\} $  defined by \eqref{eq:stru-cons}, then
$(A,\{ \cdot ,\cdot,\cdot\})$ is an $n$-dimensional  3-Pre-Lie algebra if and only if the structure constants satisfy the following conditions :
\begin{equation}\label{eq:multi-111}
\sum_{t,l}^n\big(C_{ijk}^tC_{sut}^l-(C_{sui}^t+C_{uis}^t+C_{isu}^t)C_{tjk}^l
-(C_{suj}^t+C_{ujs}^t+C_{jsu}^t) C_{itk}^l-C_{suk}^tC_{ijt}^l\big)
=0,
\end{equation}
\begin{equation}\label{eq:multi-222}
\sum_{t,l}^n\big((C_{sui}^t+C_{uis}^t+C_{isu}^t)C_{tjk}^l -C_{ijk}^tC_{sut}^l-C_{sjk}^tC_{uit}^l-C_{ujk}^tC_{ist}^l\big)=0,
\end{equation}
where$~1 \leq i,j,k,s,u,t,l\leq n$.
\end{proposition}
\begin{proof}
Select five elements $e_s$, $e_u$, $e_i$, $e_j$, $e_k$ from the basis  of a 3-Pre-Lie algebra, where $1 \leq s,u,i,j,k\leq n$.  According to \eqref{eq:multi-11} and \eqref{eq:multi-12} in the definition of 3-Pre-Lie algebra,
we can obtain the following  identities:
 \begin{eqnarray}
   \{e_s,e_u,\{e_i,e_j,e_k\}\} &=& \{[e_s,e_u,e_i]^c,e_j,e_k\}+\{e_i,[e_s,e_u,e_j]^c,e_k\} \label{eq:3-Pre-Lie 1}\\
   \nonumber&+& \{e_i,e_j,\{e_s,e_u,e_k\}\},
\end{eqnarray}
\begin{eqnarray}
   \{[e_s,e_u,e_i]^c,e_j,e_k\} &=&\{e_s,e_u,\{e_i,e_j,e_k\}\} +\{e_u,e_i,\{e_s,e_j,e_k\}\} \label{eq:3-Pre-Lie 2}\\
\mathfrak{}   \nonumber&+& \{e_i,e_s,\{e_u,e_j,e_k\}\}.
\end{eqnarray}
Calculate the left side of   \eqref{eq:3-Pre-Lie 1}, we have    $$\{e_s,e_u,\{e_i,e_j,e_k\}\}=\{e_s,e_u,\sum_t^nC_{ijk}^te_t\}=\sum_t^nC_{ijk}^t\{e_s,e_u,e_t\} =\sum_{t,l}^nC_{sut}^lC_{ijk}^te_l.$$
Calculate the right side of   \eqref{eq:3-Pre-Lie 1}, we have
\begin{eqnarray*}
  && \{[e_s,e_u,e_i]^c,e_j,e_k\}+\{e_i,[e_s,e_u,e_j]^c,e_k\}+\{e_i,e_j,\{e_s,e_u,e_k\}\}\\
  &=&\{\{e_s,e_u,e_i\},e_j,e_k\}+\{\{e_u,e_i,e_s\},e_j,e_k\}+\{\{e_i,e_s,e_u\},e_j,e_k\}\\
  &+&\{e_i,\{e_s,e_u,,e_j\},e_k\}+\{e_i,\{e_u,e_j,e_s\},e_k\}+\{e_i,\{e_j,e_s,e_u\},e_k\}\\
  &+&\{e_i,e_j,\{e_s,e_u,e_k\}\} \\
  &=&\{\sum_t^nC_{sui}^te_t,e_j,e_k\}+\{\sum_t^nC_{uis}^te_t,e_j,e_k\}+\{\sum_t^nC_{isu}^te_t,e_j,e_k\} \\
  &+&\{e_i,\sum_t^nC_{suj}^te_t,e_k\}+\{e_i,\sum_t^nC_{ujs}^te_t,e_k\}+\{e_i,\sum_t^nC_{jsu}^te_t,e_k\} \\
  &+&\{e_i,e_j,\sum_t^nC_{suk}^te_t\} \\
  &=&\sum_t^nC_{sui}^t\sum_l^nC_{tjk}^le_l+\sum_t^nC_{uis}^t\sum_l^nC_{tjk}^le_l+\sum_t^nC_{isu}^t\sum_l^nC_{tjk}^le_l \\
  &+&\sum_t^nC_{suj}^t\sum_l^nC_{itk}^le_l+\sum_t^nC_{ujs}^t\sum_l^nC_{itk}^le_l+\sum_t^nC_{jsu}^t\sum_l^nC_{itk}^le_l\\
  &+&\sum_t^nC_{suk}^t\sum_l^nC_{ijt}^le_l  \\
  &=&\sum_t^n(C_{sui}^t+C_{uis}^t+C_{isu}^t)\sum_l^nC_{tjk}^le_l +\sum_t^n(C_{suj}^t+C_{ujs}^t+C_{jsu}^t)\sum_l^nC_{itk}^le_l\\
  &+&\sum_t^nC_{suk}^t\sum_l^nC_{ijt}^le_l.
\end{eqnarray*}
The identity \eqref{eq:3-Pre-Lie 1} holds if and only if
\begin{eqnarray}\label{eq:3-Pre-Lie 111111}
  &&\sum_{t,l}^nC_{ijk}^tC_{sut}^le_l = \sum_{t,l}^n\big((C_{sui}^t+C_{uis}^t+C_{isu}^t)C_{tjk}^l \nonumber\\
 &&\quad+(C_{suj}^t+C_{ujs}^t+C_{jsu}^t)C_{itk}^l+C_{suk}^tC_{ijt}^l\big)e_l.
\end{eqnarray}
 Comparing the coefficients in \eqref{eq:3-Pre-Lie 111111}, we obtain \eqref{eq:multi-111}.

Calculate the left side of   \eqref{eq:3-Pre-Lie 2}, we have
 \begin{eqnarray*}
  &&\{\sum_t^nC_{sui}^te_t,e_j,e_k\}+\{\sum_t^nC_{uis}^te_t,e_j,e_k\}+\{\sum_t^nC_{isu}^te_t,e_j,e_k\} \\
 &=&\sum_t^nC_{sui}^t\{e_t,e_j,e_k\}+\sum_t^nC_{uis}^t\{e_t,e_j,e_k\}+\sum_t^nC_{isu}^t\{e_t,e_j,e_k\}\\
 &=&\sum_t^nC_{sui}^t\sum_l^nC_{tjk}^le_l+\sum_t^nC_{uis}^t\sum_l^nC_{tjk}^le_l +\sum_t^nC_{isu}^t\sum_l^nC_{tjk}^le_l\\
 &=&\sum_{t,l}^nC_{tjk}^l\big(C_{sui}^t+C_{uis}^t +C_{isu}^t \big)e_l.
 \end{eqnarray*}
Calculate the right side of   \eqref{eq:3-Pre-Lie 2}, we obtain
 \begin{eqnarray*}
  &&\{e_s,e_u,\{e_i,e_j,e_k\}\} +\{e_u,e_i,\{e_s,e_j,e_k\}\}+ \{e_i,e_s,\{e_u,e_j,e_k\}\}\\
  \nonumber&=&\{e_s,e_u,\sum_t^nC_{ijk}^te_t\} +\{e_u,e_i,\sum_t^nC_{sjk}^te_t\}+ \{e_i,e_s,\sum_t^nC_{ujk}^te_t\}\\
  \nonumber&=&\sum_t^nC_{ijk}^t\{e_s,e_u,e_t\} +\sum_t^nC_{sjk}^t\{e_u,e_i,e_t\}+ \sum_t^nC_{ujk}^t\{e_i,e_s,e_t\}\\
  \nonumber&=&\sum_t^nC_{ijk}^t\sum_l^nC_{sut}^le_l+\sum_t^nC_{sjk}^t\sum_l^nC_{uit}^le_l+\sum_t^nC_{ujk}^t\sum_l^nC_{ist}^le_l\\
  \nonumber&=&\sum_{t,l}^n(C_{ijk}^tC_{sut}^l+C_{sjk}^tC_{uit}^l+C_{ujk}^tC_{ist}^l)e_l.
\end{eqnarray*}
The identity \eqref{eq:3-Pre-Lie 2} holds if and only if
\begin{eqnarray}
  && \sum_{t,l}^nC_{tjk}^l\big(C_{sui}^t+C_{uis}^t +C_{isu}^t \big)e_l  \nonumber\\
 &=& \sum_{t,l}^n(C_{ijk}^tC_{sut}^l+C_{sjk}^tC_{uit}^l+C_{ujk}^tC_{ist}^l)e_l.\label{eq:3-Pre-Lie 2222}
\end{eqnarray}
 Comparing the coefficients in \eqref{eq:3-Pre-Lie 2222}, we obtain \eqref{eq:multi-222}. Therefore, the conclusion holds.
\end{proof}

\begin{theorem}\label{1111}
Let  $(A, \{ \cdot ,\cdot,\cdot\} )$ be a two-dimensional 3-Pre-Lie algebra with a basis $\{e_1,e_2\}$. Then the two-dimensional 3-Pre-Lie algebra is trivial.
\end{theorem}
\begin{proof}
Choose $e_s, e_u, e_i, e_j$ and $e_k$ in \eqref{eq:3-Pre-Lie 1} and \eqref{eq:3-Pre-Lie 2} as $e_1$ or $e_2$. We have the following five cases.

 {\bf Case 1:} When all $5$ elements are $e_1$ or $e_2$, then  \eqref{eq:3-Pre-Lie 1} and \eqref{eq:3-Pre-Lie 2}   hold.

  {\bf Case 2:} When four of the five elements are $e_1$ and one is $e_2$, then  \eqref{eq:3-Pre-Lie 1} and \eqref{eq:3-Pre-Lie 2}  have five cases, respectively.
Set the first element be $e_2$ others be $e_1$,  then \eqref{eq:3-Pre-Lie 1} and \eqref{eq:3-Pre-Lie 2} can be simplified to
       \begin{eqnarray}
            \{e_2,e_1,\{e_1,e_1,e_1\}\} &=& \{[e_2,e_1,e_1]^c,e_1,e_1\}+\{e_1,[e_2,e_1,e_1]^c,e_1\} \label{2a1}\\
            \nonumber&+& \{e_1,e_1,\{e_2,e_1,e_1\}\},\\
            \{[e_2,e_1,e_1]^c,e_1,e_1\}&=&\{e_2,e_1,\{e_1,e_1,e_1\}\}+\{e_1,e_1,\{e_2,e_1,e_1\}\}\label{2a2}\\
            \nonumber&+&\{e_1,e_2,\{e_1,e_1,e_1\}\}.
            \end{eqnarray}
By \eqref{eq:skew-sym} in the definition of 3-Pre-Lie algebra and the skew-symmetry of $[\cdot,\cdot,\cdot]^c$, $\eqref{2a1}$ and $\eqref{2a2}$  holds.
The other four cases can be proved in the same way.

 {\bf Case 3:}:When three of the five elements are $e_1$ and two are $e_2$, then  \eqref{eq:3-Pre-Lie 1} and \eqref{eq:3-Pre-Lie 2}  have ten cases, respectively.
Set the first and the second element be $e_2$, others be $e_1$.  \eqref{eq:3-Pre-Lie 1} and \eqref{eq:3-Pre-Lie 2} can be simplified to
     \begin{eqnarray}
           \{e_2,e_2,\{e_1,e_1,e_1\}\} &=& \{[e_2,e_2,e_1]^c,e_1,e_1\}+\{e_1,[e_2,e_2,e_1]^c,e_1\} \label{3a1}\\
           \nonumber&+&\{e_1,e_1,\{e_2,e_2,e_1\}\},\\
           \{[e_2,e_2,e_1]^c,e_1,e_1\}&=&\{e_2,e_2,\{e_1,e_1,e_1\}\}+\{e_2,e_1,\{e_2,e_1,e_1\}\}\label{3a2}\\
           \nonumber &+&\{e_1,e_2,\{e_2,e_1,e_1\}\}.
           \end{eqnarray}
By \eqref{eq:skew-sym} in the definition of 3-Pre-Lie algebra and skew-symmetry of $[\cdot,\cdot,\cdot]^c$, $\eqref{3a1}$ and $\eqref{3a2}$  holds. The other nine cases can be proved in the same way.

{\bf Case 4:}:When three of the five elements are $e_2$ and two are $e_1$, there are ten cases. By  Case 3, \eqref{eq:3-Pre-Lie 1} and \eqref{eq:3-Pre-Lie 2}  hold.

  {\bf Case 5:}:When four of the five elements are $e_2$ and one is $e_1$, there are five cases. By  Case 2, \eqref{eq:3-Pre-Lie 1} and \eqref{eq:3-Pre-Lie 2}  hold.

Therefore, the conclusion holds.
\end{proof}

\begin{remark}Using the structure constants in Proposition \ref{4.1} to classify  three-dimensional 3-Pre-Lie algebra, we attain a cubic equation set which is composed of 56 equations with 27 parameters. It is complicated by direct computation. But we can get some examples of three-dimensional 3-Pre-Lie algebra by $\mathcal{O}$-operators on 3-Lie algebra.
\end{remark}

By Proposition \ref{O-3preLie} and Remark \ref{rmk:O-and-3-pre}, 3-Pre-Lie algebra could be induced by $\mathcal{O}$-operators.  Thus, we can obtain thirty-one examples of three-dimensional 3-Pre-Lie algebras from $O_i$  in  Theorem \ref{O-operqtor}.

\begin{theorem}\label{3-pre-Lie example}
Let $B_{O_i}$
denote the 3-Pre-Lie algebra induced
by $O_i$ in  Theorem \ref{O-operqtor}.
Then $O_i$ induce thirty-one 3-Pre-Lie algebras, and the non-zero multiplication tables of these 3-Pre-Lie algebras are given as follows:
\\
\vspace{0.5cm}
($B_{O_1})\left\{\begin{aligned}
&\{e_2,e_3,e_1\}=(a_{22}a_{33}-a_{23}a_{32})e_1,\\
&\{e_2,e_3,e_2\}=(a_{23}a_{31}-a_{21}a_{33})e_1,\\
&\{e_1,e_3,e_3\}=(a_{21}a_{32}-a_{22}a_{31})e_1;
\end{aligned} \right.$ \quad
($B_{O_2})\left\{\begin{aligned}&\{e_1,e_2,e_2\}=-a_{11}a_{23}e_1,\\
&\{e_1,e_2,e_3\}=-a_{11}a_{33}e_1,\\
&\{e_1,e_3,e_2\}=-a_{11}a_{33}e_1,\\
&\{e_1,e_3,e_3\}=a_{11}a_{32}e_1,\\
&\{e_2,e_3,e_1\}=-(a_{23}a_{32}+a_{33}^2)e_1,\\
&\{e_2,e_3,e_2\}=(a_{23}a_{31}-a_{21}a_{33})e_1,\\
&\{e_2,e_3,e_3\}=(a_{21}a_{32}+a_{33}a_{31})e_1
;\end{aligned} \right.$\\
\vspace{0.5cm}
($B_{O_3})\left\{\begin{aligned}&\{e_1,e_2,e_2\}=a_{13}a_{21}e_1,\\
&\{e_2,e_3,e_2\}=-a_{21}a_{33}e_1;\end{aligned} \right.$ \quad\quad\quad\quad\quad
($B_{O_4})\left\{\begin{aligned}
&\{e_1,e_3,e_1\}=-a_{13}e_1,\\
&\{e_1,e_3,e_1\}=-a_{23}e_1,\\
&\{e_2,e_3,e_1\}=-a_{23}e_1,\\
&\{e_2,e_3,e_2\}=-\dfrac{a_{23}^2}{a_{13}}e_1
;
\end{aligned} \right.$ \\
\vspace{0.5cm}
($B_{O_5})\left\{\begin{aligned}&\{e_1,e_2,e_1\}=-a_{13}a_{22}e_1,\\
&\{e_1,e_3,e_1\}=-a_{13}a_{32}e_1
;\end{aligned} \right.$\quad\quad\quad\quad\, \,\,\,
($B_{O_6})\left\{\begin{aligned}&\{e_1,e_2,e_1\}=-a_{13}a_{22}e_1,\\
&\{e_1,e_2,e_2\}=a_{13}a_{21}e_1
;\end{aligned} \right.$ \\
\vspace{0.5cm}
($B_{O_7})\left\{\begin{aligned}&\{e_1,e_2,e_1\}=-a_{13}e_1,\\
&\{e_1,e_2,e_2\}=a_{13}a_{21}e_1,\\
&\{e_1,e_3,e_1\}=-\dfrac{a_{13}}{a_{21}}e_1,\\
&\{e_1,e_3,e_2\}=a_{13}e_1,\\
&\{e_2,e_3,e_1\}=a_{13}e_1,\\
&\{e_2,e_3,e_2\}=-a_{13}a_{21}e_1
;\end{aligned} \right.$ \quad
($B_{O_8})\left\{\begin{aligned}&\{e_1,e_2,e_1\}=-a_{22}e_1,\\
&\{e_1,e_2,e_2\}=(a_{21}-a_{33})e_1,\\
&\{e_1,e_2,e_3\}=a_{22}e_1,\\
&\{e_2,e_3,e_1\}=-a_{22}e_1,\\
&\{e_2,e_3,e_2\}=(a_{21}-a_{33})e_1,\\
&\{e_2,e_3,e_3\}=a_{22}e_1;
\end{aligned} \right.$ \\
\vspace{0.5cm}
($B_{O_9})\left\{\begin{aligned}&\{e_1,e_2,e_1\}=-a_{22}e_1,\\
&\{e_1,e_2,e_3\}=a_{22}e_1,\\
&\{e_1,e_3,e_1\}=-a_{32}e_1,\\
&\{e_1,e_3,e_3\}=a_{32}e_1,\\
&\{e_2,e_3,e_1\}=-a_{22}e_1,\\
&\{e_2,e_3,e_3\}=a_{22}e_1;
\end{aligned} \right.$ \quad\quad
($B_{O_{10}})\left\{\begin{aligned}
&\{e_1,e_3,e_3\}=-a_{12}a_{31}e_1,\\
&\{e_2,e_3,e_3\}=-a_{22}a_{31}e_1;
\end{aligned} \right.$ \\
\vspace{0.5cm}
($B_{O_{11}})\left\{\begin{aligned}
&\{e_1,e_3,e_1\}=a_{12}a_{33}e_1,\\
&\{e_1,e_3,e_3\}=a_{12}a_{31}e_1;
\end{aligned} \right.$ \quad\quad\quad\quad\quad\,
($B_{O_{12}})\left\{\begin{aligned}
&\{e_1,e_2,e_1\}=a_{12}e_1,\\
&\{e_1,e_3,e_1\}=a_{12}a_{33}e_1;
\end{aligned} \right.$ \\
\vspace{0.5cm}
($B_{O_{13}})\left\{\begin{aligned}
&\{e_1,e_2,e_1\}=a_{12}e_1,\\
&\{e_1,e_2,e_3\}=-a_{12}e_1,\\
&\{e_1,e_3,e_1\}=a_{12}a_{31}e_1,\\
&\{e_1,e_3,e_3\}=-a_{12}a_{31}e_1,\\
&\{e_2,e_3,e_1\}=a_{12}e_1,\\
&\{e_2,e_3,e_3\}=-a_{12}e_1;
\end{aligned} \right.$\quad\quad\quad\quad\quad
($B_{O_{14}})\left\{\begin{aligned}
&\{e_1,e_2,e_2\}=a_{21}e_1,\\
&\{e_1,e_2,e_3\}=-a_{21}e_1,\\
&\{e_1,e_3,e_2\}=-a_{21}e_1,\\
&\{e_1,e_3,e_3\}=-a_{21}e_1,\\
&\{e_2,e_3,e_2\}=-(a_{22}a_{21}+a_{32}a_{21})e_1,\\
&\{e_2,e_3,e_3\}=(a_{22}a_{21}+a_{32}a_{21})e_1;
\end{aligned} \right.$ \\
\vspace{0.5cm}
($B_{O_{15}})\left\{\begin{aligned}
&\{e_1,e_2,e_1\}=(a_{23}-a_{22})e_1,\\
&\{e_1,e_2,e_2\}=a_{32}(a_{22}-a_{23})e_1,\\
&\{e_1,e_2,e_3\}=a_{32}(a_{23}-a_{22})e_1,\\
&\{e_2,e_3,e_1\}=(a_{22}a_{32}-a_{23}a_{32})e_1,\\
&\{e_2,e_3,e_2\}=a_{32}^2(a_{23}-a_{22})e_1,\\
&\{e_2,e_3,e_3\}=a_{32}^2(a_{22}-a_{23})e_1;
\end{aligned} \right.$ \,\,\,
($B_{O_{16}})\left\{\begin{aligned}
&\{e_1,e_3,e_1\}=(a_{33}-a_{32})e_1,\\
&\{e_1,e_3,e_2\}=a_{22}(a_{33}-a_{32})e_1,\\
&\{e_1,e_3,e_3\}=-a_{22}(a_{33}-a_{32})e_1,\\
&\{e_2,e_3,e_1\}=(a_{22}a_{33}-a_{22}a_{32})e_1,\\
&\{e_2,e_3,e_2\}=a_{22}^2(a_{33}-a_{32})e_1,\\
&\{e_2,e_3,e_3\}=-a_{22}^2(a_{33}-a_{32})e_1;
\end{aligned} \right.$ \\
\vspace{0.5cm}
($B_{O_{17}})\left\{\begin{aligned}
&\{e_1,e_2,e_1\}=a_{12}a_{23}e_1,\\
&\{e_1,e_3,e_1\}=a_{12}a_{33}e_1;
\end{aligned} \right.$ \quad \quad\quad\quad\quad\,
($B_{O_{18}})\left\{\begin{aligned}
&\{e_1,e_2,e_1\}=(a_{23}-a_{22})e_1,\\
&\{e_1,e_3,e_1\}=\dfrac{a_{23}a_{32}-a_{23}a_{32}}{a_{22}}e_1;
\end{aligned} \right.$ \\
\vspace{0.5cm}
($B_{O_{19}})\left\{\begin{aligned}
&\{e_1,e_2,e_1\}=e_1,\\
&\{e_1,e_2,e_2\}=a_{21}e_1,\\
&\{e_1,e_2,e_3\}=-a_{21}e_1,\\
&\{e_1,e_3,e_1\}=e_1,\\
&\{e_1,e_3,e_2\}=a_{21}e_1,\\
&\{e_1,e_3,e_3\}=-a_{21}e_1;
\end{aligned} \right.$ \quad\quad\quad\quad\quad\quad
($B_{O_{20}})\left\{\begin{aligned}
&\{e_1,e_3,e_3\}=(a_{11}a_{32}-a_{12}a_{31})e_1;
\end{aligned} \right.$ \\
\vspace{0.5cm}
($B_{O_{21}})\left\{\begin{aligned}
&\{e_1,e_2,e_1\}=a_{12}a_{23}e_1,\\
&\{e_1,e_2,e_2\}=-a_{11}a_{23}e_1,\\
&\{e_1,e_2,e_3\}=-a_{11}a_{33}e_1,\\
&\{e_1,e_3,e_1\}=a_{12}a_{33}e_1,\\
&\{e_1,e_3,e_2\}=-a_{11}a_{33}e_1,\\
&\{e_1,e_3,e_3\}=-\dfrac{a_{11}a_{33}^2}{a_{23}}e_1;
\end{aligned} \right.$ \quad\quad\quad\quad
($B_{O_{22}})\left\{\begin{aligned}
&\{e_1,e_3,e_1\}=-a_{33}^2e_1,\\
&\{e_1,e_3,e_2\}=-a_{33}^2e_1,\\
&\{e_1,e_3,e_3\}=a_{33}(a_{32}+a_{33}a_{31})e_1,\\
&\{e_2,e_3,e_1\}=-a_{33}^2e_1,\\
&\{e_2,e_3,e_2\}=-a_{33}e_1,\\
&\{e_2,e_3,e_3\}=(a_{32}+a_{33}a_{31})e_1;
\end{aligned} \right.$ \\
\vspace{0.5cm}
($B_{O_{23}})\left\{\begin{aligned}
&\{e_1,e_2,e_1\}=e_1,\\
&\{e_1,e_2,e_2\}=-a_{11}e_1,\\
&\{e_1,e_2,e_3\}=-(a_{11}+1)e_1,\\
&\{e_1,e_3,e_1\}=e_1,\\
&\{e_1,e_3,e_2\}=-a_{11}e_1,\\
&\{e_1,e_3,e_3\}=-(a_{11}+1)e_1,\\
&\{e_2,e_3,e_1\}=e_1,\\
&\{e_2,e_3,e_2\}=-a_{11}e_1,\\
&\{e_2,e_3,e_3\}=-(a_{11}+1)e_1;
\end{aligned} \right.$ \quad\quad\quad
($B_{O_{24}})\left\{\begin{aligned}
&\{e_1,e_3,e_1\}=-e_1,\\
&\{e_1,e_3,e_2\}=e_1,\\
&\{e_1,e_3,e_3\}=(a_{32}-a_{31})e_1,\\
&\{e_2,e_3,e_1\}=e_1,\\
&\{e_2,e_3,e_2\}=-e_1,\\
&\{e_2,e_3,e_3\}=(a_{31}-a_{32})e_1;
\end{aligned} \right.$ \\
\vspace{0.5cm}
($B_{O_{25}})\left\{\begin{aligned}
&\{e_1,e_2,e_1\}=a_{23}e_1,\\
&\{e_1,e_2,e_2\}=-a_{23}e_1,\\
&\{e_1,e_2,e_3\}=(2+a_{23}a_{31})e_1,\\
&\{e_1,e_3,e_1\}=-e_1,\\
&\{e_1,e_3,e_2\}=e_1,\\
&\{e_1,e_3,e_3\}=-(\dfrac{2}{a_{23}}+a_{31})e_1,\\
&\{e_2,e_3,e_1\}=-(1+a_{23}a_{31})e_1,\\
&\{e_2,e_3,e_2\}=(a_{23}a_{31}+1)e_1,\\
&\{e_2,e_3,e_3\}=-(\dfrac{2}{a_{23}}+3a_{31}+a_{23}a_{31}^2)e_1;
\end{aligned} \right.$
($B_{O_{26}})\left\{\begin{aligned}
&\{e_1,e_2,e_1\}=(a_{23}+a_{33})e_1,\\
&\{e_1,e_2,e_2\}=-a_{11}a_{23}e_1,\\
&\{e_1,e_2,e_3\}=-a_{11}a_{33}e_1,\\
&\{e_1,e_3,e_1\}=\dfrac{a_{33}(a_{23}+a_{33})}{a_{23}}e_1,\\
&\{e_1,e_3,e_2\}=-a_{11}a_{33}e_1,\\
&\{e_1,e_3,e_3\}=-\dfrac{a_{11}a_{33}^2}{a_{23}}e_1;
\end{aligned} \right.$ \\
\vspace{0.5cm}
($B_{O_{27}})\left\{\begin{aligned}
&\{e_1,e_2,e_1\}=-e_1,\\
&\{e_1,e_2,e_3\}=e_1,\\
&\{e_1,e_3,e_1\}=-(1+a_{32})e_1,\\
&\{e_1,e_3,e_3\}=(1+a_{32})e_1,\\
&\{e_2,e_3,e_1\}=-e_1,\\
&\{e_2,e_3,e_3\}=e_1;
\end{aligned} \right.$ \quad\quad\quad\quad\quad
($B_{O_{28}})\left\{\begin{aligned}
&\{e_1,e_2,e_1\}=-(a_{33}+a_{22}-1)e_1,\\
&\{e_1,e_2,e_2\}=-(a_{33}-1)e_1,\\
&\{e_1,e_2,e_3\}=-a_{22}e_1,\\
&\{e_1,e_3,e_1\}=(a_{33}+a_{22}-1)e_1,\\
&\{e_1,e_3,e_2\}=(a_{33}-1)e_1,\\
&\{e_1,e_3,e_3\}=a_{22}e_1,\\
&\{e_2,e_3,e_1\}=(a_{33}+a_{22}-1)e_1,\\
&\{e_2,e_3,e_2\}=(a_{33}-1)e_1,\\
&\{e_2,e_3,e_3\}=a_{22}e_1;
\end{aligned} \right.$ \\
\vspace{0.5cm}
($B_{O_{29}})\left\{\begin{aligned}
&\{e_1,e_2,e_1\}=(a_{23}+a_{33})e_1,\\
&\{e_1,e_2,e_2\}=a_{33}(a_{23}-a_{33})e_1,\\
&\{e_1,e_2,e_3\}=a_{33}(a_{33}-a_{23})e_1,\\
&\{e_2,e_3,e_1\}=a_{33}(a_{23}-a_{33})e_1,\\
&\{e_2,e_3,e_2\}=a_{33}^2(a_{33}-a_{23})e_1,\\
&\{e_2,e_3,e_3\}=a_{33}^2(a_{33}-a_{23})e_1;
\end{aligned} \right.$\quad\quad\quad\quad
($B_{O_{30}})\left\{\begin{aligned}
&\{e_1,e_2,e_1\}=(\sqrt{3}-1)e_1,\\
&\{e_1,e_2,e_2\}=(2-\sqrt{3})e_1,\\
&\{e_1,e_2,e_3\}=-e_1,\\
&\{e_1,e_3,e_1\}=(\sqrt{3}+1)e_1,\\
&\{e_1,e_3,e_2\}=e_1,\\
&\{e_1,e_3,e_3\}=-(\sqrt{3}+2)e_1,\\
&\{e_2,e_3,e_1\}=-2e_1,\\
&\{e_2,e_3,e_2\}=(\sqrt{3}-1)e_1,\\
&\{e_2,e_3,e_3\}=-(\sqrt{3}+1)e_1;
\end{aligned} \right.$ \\
\vspace{0.5cm}
($B_{O_{31}})\left\{\begin{aligned}
&\{e_1,e_2,e_1\}=-(\sqrt{3}+1)e_1,\\
&\{e_1,e_2,e_2\}=(2+\sqrt{3})e_1,\\
&\{e_1,e_2,e_3\}=-e_1,\\
&\{e_1,e_3,e_1\}=(1-\sqrt{3})e_1,\\
&\{e_1,e_3,e_2\}=e_1,\\
&\{e_1,e_3,e_3\}=(\sqrt{3}-2)e_1,\\
&\{e_2,e_3,e_1\}=2e_1,\\
&\{e_2,e_3,e_2\}=-(\sqrt{3}+1)e_1,\\
&\{e_2,e_3,e_3\}=(\sqrt{3}-1)e_1.
\end{aligned} \right.$
\end{theorem}
\begin{proof}
The first $\mathcal{O}$-operator   on three-dimensional 3-Lie algebra $(A,[\cdot,\cdot,\cdot])$ in Theorem \ref{O-operqtor} is
\begin{eqnarray*}
 Te_1 = 0,~~
  Te_2 =a_{21}e_1+a_{22}e_2+a_{23}e_3,~~
   Te_3=a_{31}e_1+a_{32}e_2+a_{33}e_3.
\end{eqnarray*}
By Proposition \ref{prop:3d-3Lie} and Remark \ref{rmk:O-and-3-pre}, we have
 \begin{eqnarray*}
 \nonumber&&\{e_2,e_3,e_1\}=[Te_2,Te_3,e_1]=[a_{21}e_1+a_{22}e_2+a_{23}e_3,a_{31}e_1+a_{32}e_2+a_{33}e_3,e_1]\\
\nonumber&&=[a_{21}e_1,a_{31}e_1,e_1]+[a_{21}e_1,a_{32}e_2,e_1]+[a_{21}e_1,a_{33}e_3,e_1]\\
\nonumber&&\quad+[a_{22}e_2,a_{31}e_1,e_1]+[a_{22}e_2,a_{32}e_2,e_1]+[a_{22}e_2,a_{33}e_3,e_1]\\
\nonumber&&\quad+[a_{23}e_3,a_{31}e_1,e_1]+[a_{23}e_3,a_{32}e_2,e_1]+[a_{23}e_3,a_{33}e_3,e_1]\\
\nonumber&&=a_{21}a_{31}[e_1,e_1,e_1]+a_{21}a_{32}[e_1,e_2,e_1]+a_{21}a_{33}[e_1,e_3,e_1]\\
\nonumber&&\quad+a_{22}a_{31}[e_2,e_1,e_1]+a_{22}a_{32}[e_2,e_2,e_1]+a_{22}a_{33}[e_2,e_3,e_1]\\
\nonumber&&\quad+a_{23}a_{31}[e_3,e_1,e_1]+a_{23}a_{32}[e_3,e_2,e_1]+a_{23}a_{33}[e_3,e_3,e_1]\\
\nonumber&&=a_{22}a_{33}[e_2,e_3,e_1]+a_{23}a_{32}[e_3,e_2,e_1]=(a_{22}a_{33}-a_{23}a_{32})e_1.\label{B-O1}
\end{eqnarray*}
Calculating other multiplications in $B_{O_1}$,  we obtain
$$\{e_2,e_3,e_2\}=(a_{23}a_{31}-a_{21}a_{33})e_1,\quad  \{e_1,e_3,e_3\}=(a_{21}a_{32}-a_{22}a_{31})e_1.$$
Therefore, we obtain  $B_{O_1}$.

Other thirty examples of 3-Pre-Lie algebras can be calculated by the same method. This finishes the proof.
\end{proof}

\section{Solutions of 3-Lie  CYBE induced by  $\mathcal{O}$-operators on 3-dimensional complex 3-Lie algebra}\lb{intro}
In this section, we introduced the
 semi-direct product  3-Lie algebra  $A\ltimes_{\ad^*} A^*$ on complex 3-Lie algebra.  Then we obtain the   skew-symmetric solution  of 3-Lie  CYBE in  $A\ltimes_{\ad^*} A^*$ by the $\mathcal{O}$-operators given in   Theorem \ref{O-operqtor}.
\begin{lemma}\label{semi-direct product} {\rm{\cite{Kasymov}}}
Let $(A,[\cdot,\cdot,\cdot])$ be a 3-Lie algebra, V be a vector space and $ \rho:\otimes^2A\rightarrow \gl(V)$ be a skew-symmetric linear map. Then $(V,\rho)$ is a representation of $A$ if and only if there is a 3-Lie algebra structure  on the direct sum  $A\oplus V$ of vector spaces, defined by
\begin{eqnarray}
  [x_1+v_1,x_2+v_2,x_3+v_3]_{A\oplus V} &=&[x_1,x_2,x_3]+\rho(x_1,x_2)v_3+\rho(x_3,x_1)v_2\\
  \nonumber&+&\rho(x_2,x_3)v_1,
\end{eqnarray}
for $x_i\in A, v_i \in V, 1\leq i \leq3 $. We denote this semi-direct product 3-Lie algebra by $A\ltimes_\rho V.$
\end{lemma}
\begin{remark}\label{rema5.2}
Let V be  A and $ \rho:\otimes^2A\rightarrow \gl(V)$ be   $\ad^*:\otimes^2A\rightarrow \gl(A^*)$.  Then there is a 3-Lie algebra structure on  $A\oplus A^*$, defined by
\begin{eqnarray}\label{eq:dir-sum}
  [x_1+v_1,x_2+v_2,x_3+v_3]_{A\oplus  A^*} &=&[x_1,x_2,x_3]+\ad^*(x_1,x_2)v_3+\ad^*(x_3,x_1)v_2\\
  \nonumber &+&\ad^*(x_2,x_3)v_1,
\end{eqnarray}
for $x_i\in A, v_i \in A^*, 1\leq i \leq3 $. We denote this semi-direct product 3-Lie algebra by $A\ltimes_{\ad^*} A^*.$
\end{remark}

\begin{proposition}\label{pro5.4}
Let $\ad^*$ be the coadjoint representation of $(A,\ad)$,  $\{e_1,e_2,e_3\}$ be the basis of $A$, and $\{e_1^*,e_2^*,e_3^*\}$ be the dual basis.  Then the nonzero multiplications of the 6-dimensional 3-Lie algebra $A\ltimes_{\ad^*} A^*$  are as follows:
$$[e_1,e_2,e_3]_{A\oplus  A^*}=e_1,
~~ [e_1,e_2,e_1^*]_{A\oplus  A^*}=-e_3^*,
~~ [e_1,e_3,e_1^*]_{A\oplus  A^*}=e_2^*,
~~ [e_2,e_3,e_1^*]_{A\oplus  A^*}=-e_1^*.$$
\end{proposition}
\begin{proof}
By Proposition \ref{prop:3d-3Lie} and \eqref{eq:dir-sum} we have $[e_1,e_2,e_3]_{A\oplus  A^*}=e_1$ and
\begin{eqnarray*}
  [e_1,e_2,e_1^*]_{A\oplus  A^*}&=&[e_1+0,e_2+0,0+e_1^*]_{A\oplus  A^*} \\
   &=& [e_1,e_2,0]+\ad^*(e_1,e_2)e_1^*+\ad^*(0,e_1)0+\ad^*(e_2,0)e_1\\
   &=&\ad^*(e_1,e_2)e_1^*.
\end{eqnarray*}
Set $\ad^*(e_1,e_2)e_1^*=k_1e_1^*+k_2e_2^*+k_3e_3^*$, $k_1,k_2,k_3 \in \mathbb{C}$.
On the one hand, by the definition of coadjoint representation,  we have
$$ \langle\ad^*(e_1,e_2)e_1^*,e_1\rangle=-\langle e_1^*,[e_1,e_2,e_1]\rangle=0.$$
On the other hand, by $\langle e_i^*,e_j\rangle= \delta_{ij}$, we have
$$ \langle k_1e_1^*+k_2e_2^*+k_3e_3^*,e_1\rangle= \langle k_1e_1^*,e_1\rangle=k_1\langle e_1^*,e_1\rangle=k_1.$$
 Thus we obtain $k_1=0$.

Considering the action of
 $\ad^*(e_1,e_2)e_1^*$ on $e_2$,  we have $$\langle\ad^*(e_1,e_2)e_1^*,e_2\rangle=-\langle e_1^*,[e_1,e_2,e_2]\rangle=0,$$
and
$$\langle k_1e_1^*+k_2e_2^*+k_3e_3^*,e_2\rangle= \langle k_1e_1^*,e_2\rangle=k_2\langle e_2^*,e_2\rangle=k_2.$$
Thus we have $k_2=0$.

Considering the action of
 $\ad^*(e_1,e_2)e_1^*$ on $e_3$,  we have
$$ \langle\ad^*(e_1,e_2)e_1^*,e_3\rangle=-\langle e_1^*,[e_1,e_2,e_3]\rangle=-\langle e_1^*,e_1\rangle=-1,$$
and
 $$\langle k_1e_1^*+k_2e_2^*+k_3e_3^*,e_3\rangle= \langle k_3e_3^*,e_3\rangle=k_3\langle e_3^*,e_3\rangle=k_3.$$
Thus we obtain $k_3=-1$.
Then $\ad^*(e_1,e_2)e_1^*=-e_3^*$.  Therefore, $[e_1,e_2,e_1^*]_{A\oplus  A^*}=-e_3^*$.
Similarly, we can get $[e_1,e_3,e_1^*]_{A\oplus  A^*}=e_2^*$ and
 $[e_2,e_3,e_1^*]_{A\oplus  A^*}=-e_1^*$. This finishes the proof.
\end{proof}

For any $r=x_1\otimes x_2  \in A \otimes A$, define the switching operator $\sigma_{12}:A \otimes A\rightarrow A \otimes A$ by
\begin{equation}
  \sigma_{12} ( r) =x_2 \otimes x_1.
\end{equation}
The following result is the relationship between the $\mathcal{O}$-operators on   3-Lie algebras and solutions of 3-Lie CYBE.

\begin{lemma}\label{3-Lyb}{\rm(\cite{Bai-2})}
Let $(A,[\cdot,\cdot,\cdot])$  be a 3-Lie algebra, and let $T:A\rightarrow A$ be a linear map.  If $\bar{T}\in A^*\otimes A$ is a tensor defined by
\begin{equation}
 \bar{T} (a,\xi) = \langle\xi,Ta\rangle, \quad \forall~ a\in A,\xi \in A^*,
\end{equation}
then $T$ is an  $\mathcal{O}$-operator if and only if
$
  r=\bar{T}-\sigma_{12} ( \bar{T})
$
is a skew-symmetric solution of the 3-Lie CYBE in  the semi-direct product 3-Lie algebra  $A\ltimes_{\ad^* }A^*.$
\end{lemma}

\begin{theorem}\label{thm:3-CYBE}
Let $\{e_1,e_2\,e_3\}$ be a basis of A, and $\{e_1^*,e_2^*,e_3^*\}$ be the dual basis. Then we can obtain the following thirty-one skew-symmetric solutions of 3-Lie  CYBE in  $A\ltimes_{\ad^*} A^*$ by the $\mathcal{O}$-operators on $(A, [ \cdot ,\cdot,\cdot] )$ given in Theorem \ref{O-operqtor} :
\begin{eqnarray*}
 r_1 &=& e_2^*\otimes(a_{21}e_1+a_{22}e_2+a_{23}e_3)-(a_{21}e_1+a_{22}e_2+a_{23}e_3)\otimes e_2^*\\
          &+&e_3^*\otimes(a_{31}e_1+a_{32}e_2+a_{33}e_3)-(a_{31}e_1+a_{32}e_2+a_{33}e_3)\otimes e_3^*,\\
    r_2 &=& e_1^*\otimes a_{11}e_1-a_{11}e_1\otimes e_1^*\\
          &+& e_2^*\otimes(a_{21}e_1-a_{33}e_2+a_{23}e_3)-(a_{21}e_1-a_{33}e_2+a_{23}e_3)\otimes e_2^*\\
          &+&e_3^*\otimes(a_{31}e_1+a_{32}e_2+a_{33}e_3)-(a_{31}e_1+a_{32}e_2+a_{33}e_3)\otimes e_3^*,\\
    r_3 &=& e_1^*\otimes a_{13}e_3-a_{13}e_3\otimes e_1^*\\
          &+& e_2^*\otimes(a_{21}e_1+a_{23}e_3)-(a_{21}e_1+a_{23}e_3)\otimes e_2^*\\
          &+&e_3^*\otimes a_{33}e_3-a_{33}e_3\otimes e_3^*,\\
     r_4 &=& e_1^*\otimes a_{13}e_3-a_{13}e_3\otimes e_1^*\\
           &+& e_2^*\otimes a_{23}e_3-a_{23}e_3\otimes e_2^*\\
          &+&e_3^*\otimes (-\dfrac{a_{23}}{a_{13}}e_1+e_2+a_{33}e_3)-(-\dfrac{a_{23}}{a_{13}}e_1+e_2+a_{33}e_3)\otimes e_3^*,\\
     r_5 &=& e_1^*\otimes a_{13}e_3-a_{13}e_3\otimes e_1^*\\
           &+&e_2^*\otimes(a_{22}e_2+a_{23}e_3)-(a_{22}e_2+a_{23}e_3)\otimes e_2^*\\
           &+&e_3^*\otimes (a_{32}e_2+\dfrac{a_{23}a_{32}}{a_{22}}e_3)-(a_{32}e_2+\dfrac{a_{23}a_{32}}{a_{22}}e_3)\otimes e_3^* ,\\
     r_6 &=& e_1^*\otimes a_{13}e_3-a_{13}e_3\otimes e_1^*\\
           &+&e_2^*\otimes(a_{21}e_1+a_{22}e_2+a_{23}e_3)-(a_{21}e_1+a_{22}e_2+a_{23}e_3)\otimes e_2^*,\\
    r_7 &=& e_1^*\otimes a_{13}e_3-a_{13}e_3\otimes e_1^*\\
           &+&e_2^*\otimes(a_{21}e_1+e_2+a_{23}e_3)-(a_{21}e_1+e_2+a_{23}e_3)\otimes e_2^*\\
           &+&e_3^*\otimes(e_1+\dfrac{1}{a_{21}}e_2+(\dfrac{a_{23}+a_{21}a_{13}}{a_{21}})a_{23}e_3)-(e_1+\dfrac{1}{a_{21}}e_2+(\dfrac{a_{23}+a_{21}a_{13}}{a_{21}})a_{23}e_3)\otimes e_3^*,\\
    r_8 &=& e_1^*\otimes (e_1+e_3)-(e_1+e_3)\otimes e_1^*\\
           &+&e_2^*\otimes(a_{21}e_1+a_{22}e_2+a_{23}e_3)-(a_{21}e_1+a_{22}e_2+a_{23}e_3)\otimes e_2^*\\
           &+& e_3^*\otimes (-e_1+e_3)-(-e_1+e_3)\otimes e_3^*,\\
     r_9 &=& e_1^*\otimes (e_1+e_3)-(e_1+e_3)\otimes e_1^*\\
           &+&e_2^*\otimes a_{22}e_2-a_{22}e_2\otimes e_2^*\\
           &+& e_3^*\otimes (-e_1+a_{32}e_2-e_3)-(-e_1+a_{32}e_2-e_3)\otimes e_3^*,\\
    r_{10} &=& e_1^*\otimes a_{12}e_2-a_{12}e_2\otimes e_1^*\\
           &+& e_2^*\otimes a_{22}e_2-a_{22}e_2\otimes e_2^*\\
           &+& e_3^*\otimes (a_{31}e_1+a_{32}e_2)-(a_{31}e_1+a_{32}e_2)\otimes e_3^*,\\
     r_{11} &=& e_1^*\otimes a_{12}e_2-a_{12}e_2\otimes e_1^*\\
           &+& e_3^*\otimes (a_{31}e_1+a_{32}e_2+a_{33}e_3)-(a_{31}e_1+a_{32}e_2+a_{33}e_3)\otimes e_3^*,\\
    r_{12} &=& e_1^*\otimes a_{12}e_2-a_{12}e_2\otimes e_1^*\\
           &+& e_2^*\otimes (a_{22}e_2+e_3)-(a_{22}e_2+e_3)\otimes e_2^*\\
           &+& e_3^*\otimes (a_{22}a_{33}e_2+a_{33}e_3)-(a_{22}a_{33}e_2+a_{33}e_3)\otimes e_3^*,\\
    r_{13} &=& e_1^*\otimes a_{12}e_2-a_{12}e_2\otimes e_1^*\\
           &+& e_2^*\otimes (e_1+a_{22}e_2+e_3)-(e_1+a_{22}e_2+e_3)\otimes e_2^*\\
           &+& e_3^*\otimes (a_{31}e_1+(a_{22}a_{31}-a_{12})e_2+a_{31}e_3)-(a_{31}e_1+(a_{22}a_{31}-a_{12})e_2+a_{31}e_3)\otimes e_3^*,\\
     r_{14} &=& e_1^*\otimes(e_2+e_3)-(e_2+e_3)\otimes e_1^*\\
           &+& e_2^*\otimes (a_{21}e_1+a_{22}e_2+a_{23}e_3)-(a_{21}e_1+a_{22}e_2+a_{23}e_3)\otimes e_2^*\\
           &+& e_3^*\otimes (-a_{21}e_1+a_{32}e_2+a_{32}e_3)-(-a_{21}e_1+a_{32}e_2+a_{32}e_3)\otimes e_3^*,\\
      r_{15} &=& e_1^*\otimes(e_2+e_3)-(e_2+e_3)\otimes e_1^*\\
           &+& e_2^*\otimes (a_{32}(a_{22}-a_{33})e_1+a_{22}e_2+a_{23}e_3)-(a_{32}(a_{22}-a_{33})e_1+a_{22}e_2+a_{23}e_3)\otimes e_2^*\\
           &+& e_3^*\otimes (a_{32}e_2+a_{32}e_3)-(a_{32}e_2+a_{32}e_3)\otimes e_3^*,\\
     r_{16} &=& e_1^*\otimes(e_2+e_3)-(e_2+e_3)\otimes e_1^*\\
           &+& e_2^*\otimes (a_{22}e_2+a_{23}e_3)-(a_{22}e_2+a_{23}e_3)\otimes e_2^*\\
           &+& e_3^*\otimes (a_{22}(a_{33}-a_{32})e_1+a_{32}e_2+a_{33}e_3)-(a_{22}(a_{33}-a_{32})e_1+a_{32}e_2+a_{33}e_3)\otimes e_3^*,\\
   r_{17} &=& e_1^*\otimes(a_{12}e_2+a_{13}e_3)-(a_{12}e_2+a_{13}e_3)\otimes e_1^*\\
          &+&e_2^*\otimes a_{23}e_3-a_{23}e_3\otimes e_2^*,\\
          &+&e_3^*\otimes a_{33}e_3-a_{33}e_3\otimes e_3^*,\\
      r_{18} &=& e_1^*\otimes(e_2+e_3)-(e_2+e_3)\otimes e_1^*\\
           &+& e_2^*\otimes (a_{22}e_2+a_{23}e_3)-(a_{22}e_2+a_{23}e_3)\otimes e_2^*\\
           &+& e_3^*\otimes (a_{32}e_2+\dfrac{a_{23}a_{32}}{a_{22}}e_3)-(a_{32}e_2+\dfrac{a_{23}a_{32}}{a_{22}}e_3)\otimes e_3^*,\\
     r_{19} &=& e_1^*\otimes(e_2+e_3)-(e_2+e_3)\otimes e_1^*\\
           &+& e_2^*\otimes (a_{21}e_1+e_2+2e_3)-(a_{21}e_1+e_2+2e_3)\otimes e_2^*\\
           &+& e_3^*\otimes (a_{21}e_1+e_2+2e_3)-(a_{21}e_1+e_2+2e_3)\otimes e_3^*,\\
    r_{20} &=& e_1^*\otimes(a_{11}e_1+a_{12}e_2)-(a_{11}e_1+a_{12}e_2)\otimes e_1^*\\
                &+& e_3^*\otimes(a_{31}e_1+a_{32}e_2)-(a_{31}e_1+a_{32}e_2)\otimes e_3^*,\\
      r_{21} &=& e_1^*\otimes(a_{11}e_1+a_{12}e_2)-(a_{11}e_1+a_{12}e_2)\otimes e_1^*\\
                &+& e_2^*\otimes (-a_{33}e_2+a_{23}e_3)-(-a_{33}e_2+a_{23}e_3)\otimes e_2^*\\
                &+& e_3^*\otimes(-\dfrac{a_{33}^2}{a_{23}}e_2+a_{33}e_3)-(-\dfrac{a_{33}^2}{a_{23}}e_2+a_{33}e_3)\otimes e_3^*,\\
      r_{22} &=& e_1^*\otimes(a_{33}e_1-a_{33}^2e_2)-(a_{33}e_1-a_{33}^2e_2)\otimes e_1^*\\
                &+& e_2^*\otimes (e_1-a_{33}e_2)-(e_1-a_{33}e_2)\otimes e_2^*\\
                &+& e_3^*\otimes (a_{31}e_1+a_{32}e_2+a_{33}e_3)- (a_{31}e_1+a_{32}e_2+a_{33}e_3)\otimes e_3^*,\\
      r_{23} &=& e_1^*\otimes(a_{11}e_1+e_2)-(a_{11}e_1+e_2)\otimes e_1^*\\
                &+& e_2^*\otimes (e_1-e_2+e_3)-(e_1-e_2+e_3)\otimes e_2^*\\
                &+& e_3^*\otimes ((1-a_{11})e_1-2e_2+e_3)-((1-a_{11})e_1-2e_2+e_3)\otimes e_3^*,\\
      r_{24} &=& e_1^*\otimes(e_1+e_2)-(e_1+e_2)\otimes e_1^*\\
                &+& e_2^*\otimes (-e_1-e_2)-(-e_1-e_2)\otimes e_2^*\\
                &+& e_3^*\otimes (a_{31}e_1+a_{32}e_2-e_3)-(a_{31}e_1+a_{32}e_2-e_3)\otimes e_3^*,\\
      r_{25} &=& e_1^*\otimes(e_1+e_2)-(e_1+e_2)\otimes e_1^*\\
                &+& e_2^*\otimes (e_1+(3+a_{23}a_{31})e_2+a_{23}e_3)- (e_1+(3+a_{23}a_{31})e_2+a_{23}e_3)\otimes e_2^*\\
                &+& e_3^*\otimes (a_{31}e_1-\dfrac{2}{a_{23}}e_2-e_3)-(a_{31}e_1-\dfrac{2}{a_{23}}e_2-e_3)\otimes e_3^*,\\
      r_{26} &=& e_1^*\otimes(a_{11}e_1+e_2+e_3)-(a_{11}e_1+e_2+e_3)\otimes e_1^*\\
                &+& e_2^*\otimes (-a_{33}e_2+a_{23}e_3)- (-a_{33}e_2+a_{23}e_3)\otimes e_2^*\\
                &+& e_3^*\otimes (-\dfrac{a_{33}^2}{a_{23}}e_2+a_{33}e_3)-(-\dfrac{a_{33}^2}{a_{23}}e_2+a_{33}e_3)\otimes e_3^*,\\
      r_{27} &=& e_1^*\otimes(e_1+e_2+e_3)-(e_1+e_2+e_3)\otimes e_1^*\\
                &+& e_2^*\otimes e_2- e_2\otimes e_2^*\\
                &+& e_3^*\otimes (-e_1+a_{32}e_2-e_3)- (-e_1+a_{32}e_2-e_3)\otimes e_3^*,\\
      r_{28} &=& e_1^*\otimes(-e_1+e_2+e_3)-(-e_1+e_2+e_3)\otimes e_1^*\\
                &+& e_2^*\otimes (a_{22}e_2+(1-a_{33})e_3)-(a_{22}e_2+(1-a_{33})e_3)\otimes e_2^*\\
                &+& e_3^*\otimes (-e_1+(1-a_{22})e_2+a_{33}e_3)- (-e_1+(1-a_{22})e_2+a_{33}e_3)\otimes e_3^*,\\
      r_{29} &=& e_1^*\otimes(e_2+e_3)-(e_2+e_3)\otimes e_1^*\\
                &+& e_2^*\otimes (a_{33}(a_{23}-a_{33})e_1-a_{33}e_2+a_{23}e_3)- (a_{33}(a_{23}-a_{33})e_1-a_{33}e_2+a_{23}e_3)\otimes e_2^*\\
                &+& e_3^*\otimes(-a_{33}e_2+a_{33}e_3)-(-a_{33}e_2+a_{33}e_3)\otimes e_3^*,\\
     r_{30} &=& e_1^*\otimes(e_1+e_2+e_3)-(e_1+e_2+e_3)\otimes e_1^*\\
                &+& e_2^*\otimes (e_1+(\sqrt{3}-1)e_3)-  (e_1+(\sqrt{3}-1)e_3)\otimes e_2^*\\
                &+& e_3^*\otimes(e_1-(\sqrt{3}+1)e_2)-(e_1-(\sqrt{3}+1)e_2)\otimes e_3^*,\\
      r_{31} &=& e_1^*\otimes(e_1+e_2+e_3)-(e_1+e_2+e_3)\otimes e_1^*\\
                &+& e_2^*\otimes (e_1-(\sqrt{3}+1)e_3)- (e_1-(\sqrt{3}+1)e_3)\otimes e_2^*\\
                &+& e_3^*\otimes(e_1+(\sqrt{3}-1)e_2)-(e_1+(\sqrt{3}-1)e_2)\otimes e_3^*.
\end{eqnarray*}
\end{theorem}
\begin{proof}
Since $\{e_1,e_2,e_3\}$ is a basis of $A$, and $\{e_1^*,e_2^*,e_3^*\}$ is the dual basis. Then we have
\begin{equation}
  \bar{T} =\sum_{i=1}^3 e_i^*\otimes Te_i \in A^*\otimes A.
\end{equation}
If $T$ is an $\mathcal{O}$-operator on 3-Lie algebra, then by Lemma \ref{3-Lyb} we have
\begin{equation}
  r =e_1^*\otimes Te_1+e_2^*\otimes Te_2+e_3^*\otimes Te_3-Te_1\otimes e_1^*+Te_2\otimes e_2^*+Te_3\otimes e_3^*,
\end{equation}
is a skew-symmetric solution of the 3-Lie CYBE in  the semi-direct product 3-Lie algebra $A\ltimes_{\ad^* }A^*$.

  Using the $\mathcal{O}$-operator $O_1$  on 3-dimensional complex 3-Lie algebra
  in Theorem \ref{O-operqtor}, we have $$ Te_1=0,\quad Te_2=a_{21}e_1+a_{22}e_2+a_{23}e_3,\quad Te_3=a_{31}e_1+a_{32}e_2+a_{33}e_3.$$
 Therefore, we obtain $r_1$. Other thirty skew-symmetric solutions of 3-Lie  CYBE in  $A\ltimes_{\ad^*} A^*$ induced by the $\mathcal{O}$-operators in Theorem \ref{O-operqtor} can be calculated by the same method. This finishes the proof.
\end{proof}

\section{Proof of Theorem \ref{O-operqtor}}

 Suppose that the linear map  $T:A\rightarrow A$  is  an $\mathcal{O}$-operator  associated to $(A,\ad)$. If $(A,[ \cdot ,\cdot,\cdot] )$ is a 3-dimensional 3-Lie algebra, then  $\eqref{eq:O-operator1.1}$ can be simplified to
\begin{eqnarray}
[Te_1,Te_2,Te_3]&=&T([Te_1,Te_2,e_3]+[Te_3,Te_1,e_2]+[Te_2,Te_3,e_1]).\label{O-operator1.2}
\end{eqnarray}

By Lemma \ref{lem:O}, \eqref{eq:stru-cons2} can be simplified to
\begin{eqnarray}
  &&a_{12}(a_{11}a_{21}-a_{21}a_{33}+a_{23}a_{31})+a_{13}(a_{21}a_{32}+a_{11}a_{31}-a_{22}a_{31})\label{eq:3d-O-equ1}  \\
\nonumber  && -a_{11}^2(a_{22}+a_{33})=0,  \\
  &&a_{12}(a_{11}a_{22}-a_{12}a_{21}-a_{31}a_{13}+a_{33}a_{11}+a_{22}a_{33}-a_{23}a_{32})=0,\label{eq:3d-O-equ2} \\
  &&a_{13}(a_{11}a_{22}-a_{12}a_{21}-a_{13}a_{31}+a_{33}a_{11}+a_{22}a_{33}-a_{23}a_{32})=0.\label{eq:3d-O-equ3}
\end{eqnarray}
To solve the cubic equations \eqref{eq:3d-O-equ1}-\eqref{eq:3d-O-equ3} of nine elements $a_{ij},~~ 1\leq i,j \leq 3$,  we distinguish
the two cases depending on whether $a_{12}=0$ or not .

{\bf Case 1:} $a_{12}=0$: Then $\eqref{eq:3d-O-equ1}$ implies $a_{13}(a_{21}a_{32}+a_{11}a_{31}-a_{22}a_{31})=a_{11}^2(a_{22}+a_{33})$ . There are two subcases:$(a)a_{13}=0$; $(b)a_{13}\neq 0$.

(a) Assume  $a_{12}=0$ and $a_{13}=0$, then \eqref{eq:3d-O-equ2} and \eqref{eq:3d-O-equ3} hold. \eqref{eq:3d-O-equ1} also implies $a_{11}^2(a_{22}+a_{33})=0$. There are two subcases $(a_{1}):a_{11}=0$;  $(a_{2}):a_{11}\neq 0$.

$(a_{1})$ If $a_{11}=0$, then we get $O_1$.

$(a_{2})$ If $a_{11}\neq 0$, then we get $O_2$.

(b) Assume  $a_{12}=0$ and $a_{13}\neq 0$, then  \eqref{eq:3d-O-equ2} holds.   \eqref{eq:3d-O-equ3} implies $a_{11}a_{22}-a_{13}a_{31}+a_{33}a_{11}+a_{22}a_{33}-a_{23}a_{32}=0$. There are two subcases $(b_{1}): a_{11}=0, (b_{2}):a_{11} \neq 0.$

$(b_{1})$ If $a_{11}=0$, then \eqref{eq:3d-O-equ3} implies $a_{22}a_{33}-a_{13}a_{31}-a_{23}a_{32}=0$. And \eqref{eq:3d-O-equ1}  implies $a_{21}a_{32}-a_{22}a_{31}=0$.  There are two subcases $(b_{11}): a_{22}=0, (b_{12}):a_{22} \neq 0.$

$(b_{11})$ If $a_{22}=0$, then \eqref{eq:3d-O-equ3} implies $a_{13}a_{31}+a_{23}a_{32}=0$. Then  \eqref{eq:3d-O-equ1}  implies $a_{21}a_{32}=0$.   There are two subcases $(b_{111}): a_{32}=0, (b_{112}):a_{32} \neq 0.$

$(b_{111})$ If $a_{32}=0$, then \eqref{eq:3d-O-equ1} holds, \eqref{eq:3d-O-equ3} implies $a_{13}a_{31}=0$. There are two subcases  $(b_{1111}): a_{31}=0, (b_{1112}): a_{31}\neq 0. $

 $(b_{1111})$ If  $a_{31}=0$,  then \eqref{eq:3d-O-equ3} holds. Then we get $O_3$.

 $(b_{1112})$ If  $a_{31} \neq 0$, then  \eqref{eq:3d-O-equ3} implies $a_{13}=0$. It is in contradiction with $a_{13} \neq 0$.

$(b_{112})$ If $a_{32}\neq0$, then \eqref{eq:3d-O-equ1} implies $a_{21}=0$. Taking  $a_{32}=1$,  \eqref{eq:3d-O-equ3} implies $a_{31}=-\dfrac{a_{23}}{a_{13}}$, then we get $O_4$.

$(b_{12})$ If $a_{22} \neq 0$, then \eqref{eq:3d-O-equ3} implies $a_{22}a_{33}=a_{23}a_{32}+a_{13}a_{31}$. And \eqref{eq:3d-O-equ1}  implies $a_{21}a_{32}=a_{22}a_{31}$. There are two subcases $(b_{121}): a_{31}=0, (b_{122}):a_{31} \neq 0.$

$(b_{121})$ If $a_{31}=0$,  then \eqref{eq:3d-O-equ1} implies $a_{21}a_{32}=0$.  \eqref{eq:3d-O-equ3} implies $a_{22}a_{33}=a_{23}a_{32}$. There are two subcases $(b_{1211}): a_{21}=0, (b_{1212}):a_{21} \neq 0.$

$(b_{1211})$ If $a_{21}=0$, then \eqref{eq:3d-O-equ1} holds.   \eqref{eq:3d-O-equ3} implies $a_{22}a_{33}=a_{23}a_{32}$.
Since $a_{22} \neq 0$, then \eqref{eq:3d-O-equ3} implies $a_{33}= \dfrac{a_{23}a_{32}}{a_{22}}$,  we get $O_5$.

$(b_{1212})$ If $a_{21}\neq0$, then \eqref{eq:3d-O-equ1}  implies $a_{32}=0$, \eqref{eq:3d-O-equ3} implies $a_{22}a_{33}=0$. Since $a_{22} \neq 0$,  thus $a_{33}=0$,  we get $O_6$.

$(b_{122})$ If $a_{31}\neq0$, taking $a_{31}=1$, $a_{22}=1$, then  \eqref{eq:3d-O-equ1}  implies $a_{32}=\dfrac{1}{a_{21}}$. Then \eqref{eq:3d-O-equ3} implies $a_{33}=\dfrac{a_{23}}{a_{21}}+a_{13}$,  we get $O_7$.

$(b_{2})$ If $a_{11}\neq0$, taking $a_{11}=a_{13}=1$, \eqref{eq:3d-O-equ3} implies $a_{22}(1+a_{33})-a_{31}+a_{33}-a_{23}a_{32}=0$.  Taking $a_{31}=a_{33}=-1$, then \eqref{eq:3d-O-equ1} implies    $a_{21}a_{32}=0$.   \eqref{eq:3d-O-equ3} implies $a_{23}a_{32}=0$.
 There are two subcases  $(b_{21}): a_{32}=0, (b_{22}): a_{32}\neq 0$.

$(b_{21})$ If $a_{32}=0$, then \eqref{eq:3d-O-equ1} and \eqref{eq:3d-O-equ3}  hold.  We get $O_8$.

$(b_{22})$ If $a_{32}\neq0$, then \eqref{eq:3d-O-equ1} implies    $a_{21}=0$.  Then \eqref{eq:3d-O-equ3} implies $a_{23}=0$,  we get $O_9$.

{\bf Case 2:} $a_{12} \neq0$: then  \eqref{eq:3d-O-equ2}  and \eqref{eq:3d-O-equ3} imply $a_{11}(a_{22}+a_{33} )-a_{12}a_{21}-a_{31}a_{13}+a_{22}a_{33}-a_{23}a_{32}=0 $.  We distinguish two subcases: $(c): a_{11}=0$, and $(d): a_{11}\neq 0$.

 $(c)$ Assume $ a_{11}=0$, then  \eqref{eq:3d-O-equ2}  and \eqref{eq:3d-O-equ3} imply $a_{22}a_{33}-a_{31}a_{13}-a_{12}a_{21}-a_{23}a_{32}=0 $. Then   \eqref{eq:3d-O-equ1} implies $a_{12}(a_{23}a_{31}-a_{21}a_{33})+a_{13}(a_{21}a_{32}-a_{22}a_{31})=0.$
 There are two subcases: $(c_1): a_{13}=0$, and $(c_2): a_{13}\neq 0$.

  $(c_1)$ If $ a_{13}=0$, then  \eqref{eq:3d-O-equ3} holds.  Then \eqref{eq:3d-O-equ1} implies $a_{23}a_{31}=a_{21}a_{33},$ \eqref{eq:3d-O-equ2}  implies $a_{22}a_{33}-a_{12}a_{21}-a_{23}a_{32}=0 $.  Then there are two subcases: $(c_{11}): a_{23}=0$, and $(c_{12}): a_{23}\neq 0$.

  $(c_{11})$ If $ a_{23}=0$, then \eqref{eq:3d-O-equ1} implies $a_{21}a_{33}=0 $.  \eqref{eq:3d-O-equ2} implies $a_{22}a_{33}-a_{12}a_{21}=0 $. Then there are two subcases: $(c_{111}): a_{21}=0$, and $(c_{112}): a_{21}\neq 0$.

 $(c_{111})$ If $ a_{21}=0$,then \eqref{eq:3d-O-equ1} holds.  Then \eqref{eq:3d-O-equ2}   implies $a_{22}a_{33}=0 $. Then there are two subcases: $(c_{1111}): a_{33}=0$, and $(c_{1112}): a_{33}\neq 0$.

  $(c_{1111})$ If $ a_{33}=0$,  \eqref{eq:3d-O-equ2}  holds.  We get $O_{10}$.

$(c_{1112})$ If $ a_{33} \neq 0$, then  \eqref{eq:3d-O-equ2}  implies $a_{22}=0$. We get $O_{11}$.

 $(c_{112})$ If $ a_{21}\neq0$, then \eqref{eq:3d-O-equ1} implies $a_{33}=0$. \eqref{eq:3d-O-equ2} implies $a_{12}=0 $. It is in contradiction with $a_{12} \neq 0$.

$(c_{12})$ If $ a_{23}\neq0$, taking $ a_{23}=1$,  then \eqref{eq:3d-O-equ1} implies $a_{31}=a_{21}a_{33}$. Then   \eqref{eq:3d-O-equ2} implies $a_{22}a_{33}-a_{12}a_{21}-a_{32}=0 $. Then there are two subcases: $(c_{121}): a_{21}=0$, and $(c_{122}): a_{21}\neq 0$.

 $(c_{121})$ If $ a_{21}=0$,  then \eqref{eq:3d-O-equ1} implies $a_{31}=0$ . Then  \eqref{eq:3d-O-equ2}  implies $a_{32}=a_{22}a_{33}$. We get $O_{12}$.

  $(c_{122})$ If $ a_{21}\neq 0$, taking $ a_{21}=1$,   then \eqref{eq:3d-O-equ1} implies $a_{33}=a_{31}$.  Then   \eqref{eq:3d-O-equ2} implies $a_{32}=a_{22}a_{31}-a_{12}$,  we get $O_{13}$.

  $(c_2)$ If $ a_{13}\neq 0$, taking $ a_{12}=a_{13}=1$, then  \eqref{eq:3d-O-equ2}  and \eqref{eq:3d-O-equ3} imply $a_{22}a_{33}-a_{31}-a_{21}-a_{23}a_{32}=0 $.  Then   \eqref{eq:3d-O-equ1} implies $a_{31}(a_{23}-a_{22})+a_{21}(a_{32}-a_{33})=0.$    Then there are two subcases: $(c_{21}): a_{33}=a_{32}$, and $(c_{22}):  a_{33}\neq a_{32}$.

  $(c_{21})$ If $a_{33}=a_{32}$, then   \eqref{eq:3d-O-equ1} implies $a_{31}(a_{23}-a_{22})=0$. Then  \eqref{eq:3d-O-equ2}  and \eqref{eq:3d-O-equ3} imply $a_{22}a_{32}-a_{31}-a_{21}-a_{23}a_{32}=0 $.
 Then there are two subcases: $(c_{211}): a_{23}=a_{22}$, and $(c_{212}):  a_{23}\neq a_{22}$.

$(c_{211})$ If $a_{23}=a_{22}$, then    \eqref{eq:3d-O-equ1} holds. Then  \eqref{eq:3d-O-equ2}  and \eqref{eq:3d-O-equ3} imply $a_{31}=-a_{21}$,  we get $O_{14}$.

$(c_{212})$ If $a_{23}\neq a_{22}$,   then   \eqref{eq:3d-O-equ1} implies $a_{31}=0.$ Then  \eqref{eq:3d-O-equ2}  and \eqref{eq:3d-O-equ3} imply $a_{21}=a_{32}(a_{22}-a_{23}) $. We get $O_{15}$.

  $(c_{22})$ If $a_{23}\neq a_{22}$,  then there are two subcases: $(c_{221}): a_{23}=a_{22}$, and $(c_{222}):  a_{23}\neq a_{22}$.

  $(c_{221})$ If $a_{23}=a_{22}$, then   \eqref{eq:3d-O-equ1} implies $a_{21}=0$.  Then  \eqref{eq:3d-O-equ2}  and \eqref{eq:3d-O-equ3} imply $a_{31}=a_{22}(a_{33}-a_{32})$.  We get $O_{16}$.

$(c_{222})$ If $a_{23}\neq a_{22}$,  then   \eqref{eq:3d-O-equ1} implies $a_{31}(a_{23}-a_{22})+a_{21}(a_{32}-a_{33})=0.$    Then  \eqref{eq:3d-O-equ2}  and \eqref{eq:3d-O-equ3} imply $a_{22}a_{33}-a_{31}-a_{21}-a_{23}a_{32}=0 $.  Then there are two subcases: $(c_{2221}): a_{21}=0$, and $(c_{2222}):  a_{21}\neq0$.

$(c_{2221})$ If $a_{21}=0$, then   \eqref{eq:3d-O-equ1} implies $a_{31}=0$.  Then  \eqref{eq:3d-O-equ2}  and \eqref{eq:3d-O-equ3} imply $a_{22}a_{33}=a_{23}a_{32} $. If $a_{22}=0$, since  $a_{23}\neq a_{22}$, then   $a_{23}\neq0$. Then  \eqref{eq:3d-O-equ2}  and \eqref{eq:3d-O-equ3} imply $a_{32}=0$.  We get $O_{17}$.

 If $a_{22}\neq 0$, then  \eqref{eq:3d-O-equ2}  and \eqref{eq:3d-O-equ3} imply $a_{33}=\dfrac{a_{23}a_{32}}{a_{22}} $.  We get $O_{18}$.

 $(c_{2222})$ If $a_{21}\neq 0$,  taking $a_{23}=2$, $a_{22}=1$, $a_{33}=2$, $a_{32}=1$, then   \eqref{eq:3d-O-equ1} ,  \eqref{eq:3d-O-equ2}  and \eqref{eq:3d-O-equ3} imply $a_{31}=a_{21}$. We get $O_{19}$.

 $(d)$ Assume $ a_{11} \neq 0$, then  \eqref{eq:3d-O-equ2}  and \eqref{eq:3d-O-equ3} imply  $a_{11}a_{22}-a_{12}a_{21}-a_{13}a_{31}+a_{33}a_{11}+a_{22}a_{33}-a_{23}a_{32}=0.$  \eqref{eq:3d-O-equ1} implies $a_{12}(a_{11}a_{21}-a_{21}a_{33}+a_{23}a_{31})+a_{13}(a_{21}a_{32}+a_{11}a_{31}-a_{22}a_{31})-a_{11}^2(a_{22}+a_{33})=0$.
 Then there are two subcases: $(d_1): a_{13}=0$, and $(d_2): a_{13}\neq 0$.

$(d_1)$ If $ a_{13}=0$, then \eqref{eq:3d-O-equ1} implies $a_{12}(a_{11}a_{21}-a_{21}a_{33}+a_{23}a_{31})-a_{11}^2(a_{22}+a_{33})=0$, then  \eqref{eq:3d-O-equ2}  and \eqref{eq:3d-O-equ3} imply   $a_{11}a_{22}-a_{12}a_{21}+a_{33}a_{11}+a_{22}a_{33}-a_{23}a_{32}=0.$    Then there are two subcases: $(d_{11}): a_{22}=-a_{33}$, and $(d_{12}): a_{22} \neq -a_{33}$.

$(d_{11})$ If $ a_{22}=-a_{33}$, then \eqref{eq:3d-O-equ1} implies $a_{12}(a_{11}a_{21}-a_{21}a_{33}+a_{23}a_{31})=0$. Since  $a_{12} \neq 0$, we have $a_{21}(a_{11}-a_{33})+a_{23}a_{31}=0$.  Then  \eqref{eq:3d-O-equ2}  and \eqref{eq:3d-O-equ3} imply   $a_{12}a_{21}+a_{33}^2+a_{23}a_{32}=0.$ Then there are two subcases: $(d_{111}): a_{21}=0$, and $(d_{112}): a_{21}\neq 0$.

$(d_{111})$ If $ a_{21}=0$, then \eqref{eq:3d-O-equ1} implies $a_{23}a_{31}=0$. Then  \eqref{eq:3d-O-equ2}  and \eqref{eq:3d-O-equ3} imply  $a_{33}^2+a_{23}a_{32}=0.$

If $ a_{23}=0$,  then  \eqref{eq:3d-O-equ2}  and \eqref{eq:3d-O-equ3} imply  $a_{33}=a_{22}=0$, we get $O_{20}$.

If $a_{23}\neq0$, then \eqref{eq:3d-O-equ1} implies $a_{31}=0$.  Then  \eqref{eq:3d-O-equ2}  and \eqref{eq:3d-O-equ3} imply   
$a_{32}=-\dfrac{a_{33}^2}{a_{23}}$,   we get $O_{21}$.

$(d_{112})$ If $ a_{21} \neq 0$,  taking  $ a_{21}=1$, then \eqref{eq:3d-O-equ1} implies $a_{11}-a_{33}+a_{23}a_{31}=0$. Then  \eqref{eq:3d-O-equ2}  and \eqref{eq:3d-O-equ3} imply  $a_{12}+a_{33}^2+a_{23}a_{32}=0$.  Then there are two subcases: $(d_{1121}): a_{23}=0$, and $(d_{1122}): a_{23}\neq 0$.

$(d_{1121})$ If $ a_{23}=0$,  then \eqref{eq:3d-O-equ1} implies $a_{11}=a_{33}$. Then  \eqref{eq:3d-O-equ2}  and \eqref{eq:3d-O-equ3} imply  $a_{12}=-a_{33}^2$.  We get $O_{22}$.

$(d_{1122})$ If $ a_{23}\neq0$,  taking  $a_{12}= a_{21}=a_{23}=a_{33}=1$,
 then \eqref{eq:3d-O-equ1} implies $a_{31}=1-a_{11}$.  Then  \eqref{eq:3d-O-equ2}  and \eqref{eq:3d-O-equ3} imply  $a_{32}=-2$, we get $O_{23}$.

$(d_{12})$ If $ a_{22} \neq -a_{33}$, taking $ a_{11}=a_{12}=1$, $ a_{33}=-1$,  then \eqref{eq:3d-O-equ1} implies $2a_{21}+a_{23}a_{31}-a_{22}+1=0$,  \eqref{eq:3d-O-equ2}  and \eqref{eq:3d-O-equ3} imply  $a_{21}+a_{23}a_{32}+1=0.$    Then there are two subcases: $(d_{121}): a_{23}=0$, and $(d_{122}): a_{23}\neq 0$.

$(d_{121})$ If $a_{23}=0$, then  \eqref{eq:3d-O-equ2}  and \eqref{eq:3d-O-equ3} imply   $a_{21}=-1$.
Then \eqref{eq:3d-O-equ1} implies $a_{22}=-1$,  we get $O_{24}$.

$(d_{122})$ If $a_{23} \neq 0$, taking $a_{21}=1$, then  \eqref{eq:3d-O-equ2}  and \eqref{eq:3d-O-equ3} imply   $a_{32}=-\dfrac{2}{a_{23}}.$ Then \eqref{eq:3d-O-equ1} implies $a_{22}=3+a_{23}a_{31}$,  we get $O_{25}$.

$(d_2)$ If $ a_{13}\neq0$, taking $a_{12}=a_{13}=1$, then \eqref{eq:3d-O-equ1} implies $a_{11}a_{21}-a_{21}a_{33}+a_{23}a_{31}+a_{21}a_{32}+a_{11}a_{31}-a_{22}a_{31} -a_{11}^2(a_{22}+a_{33})=0$. Then   \eqref{eq:3d-O-equ2}  and \eqref{eq:3d-O-equ3} imply $a_{11}a_{22}-a_{21}-a_{31}+a_{33}a_{11}+a_{22}a_{33}-a_{23}a_{32}=0 $. There are two subcases: $(d_{21}): a_{21}=0$, and $(d_{22}): a_{21}\neq 0$.

$(d_{21})$ If $a_{21}=0$, then \eqref{eq:3d-O-equ1} implies $a_{23}a_{31}+a_{11}a_{31}-a_{22}a_{31} -a_{11}^2(a_{22}+a_{33})=0$.
Then   \eqref{eq:3d-O-equ2}  and \eqref{eq:3d-O-equ3} imply $a_{11}a_{22}-a_{31}+a_{33}a_{11}+a_{22}a_{33}-a_{23}a_{32}=0 $.
 Then there are two subcases: $(d_{211}): a_{22}=-a_{33}$, and $(d_{212}): a_{22} \neq -a_{33}$.

$(d_{211})$ If $ a_{22}=-a_{33}$, then \eqref{eq:3d-O-equ1}  implies $a_{23}a_{31}+a_{11}a_{31}-a_{22}a_{31} =0$. Then   \eqref{eq:3d-O-equ2}  and \eqref{eq:3d-O-equ3} imply  $-a_{31}-a_{33}^2-a_{23}a_{32}=0 $.

If $ a_{31}=0$, then \eqref{eq:3d-O-equ1} holds.  Then   \eqref{eq:3d-O-equ2}  and \eqref{eq:3d-O-equ3} imply $a_{32}=-\dfrac{a_{33}^2}{a_{23}} $, we get $O_{26}$.

If $ a_{31} \neq 0$,  taking $a_{23}=0$, then \eqref{eq:3d-O-equ1}  implies $a_{22}=a_{11}=1 $. Since $ a_{22}=-a_{33}$, thus we get $ a_{33}=-1$.  \eqref{eq:3d-O-equ2}  and \eqref{eq:3d-O-equ3} imply $a_{31}=-a_{33}^2=-1 $, then we get $O_{27}$.

$(d_{212})$ If $ a_{22} \neq-a_{33}$, taking $a_{11}=-1$, then \eqref{eq:3d-O-equ1} implies $a_{31}(a_{23}-a_{22}-1) -(a_{22}+a_{33})=0$.
Then   \eqref{eq:3d-O-equ2}  and \eqref{eq:3d-O-equ3} imply  $-a_{22}-a_{31}-a_{33}+a_{22}a_{33}-a_{23}a_{32}=0 $.

If $ a_{31}=0$, \eqref{eq:3d-O-equ1} implies $a_{22}=-a_{33}$, it  is in contradiction with $ a_{22} \neq-a_{33}$.

If $ a_{31} \neq 0$,  taking $a_{31}=-1$,  then \eqref{eq:3d-O-equ1} implies $a_{23}=1-a_{33}$. Then   \eqref{eq:3d-O-equ2}  and \eqref{eq:3d-O-equ3} imply  $a_{32}=1-a_{22} $, we get $O_{28}$.

$(d_{22})$ If $a_{21}\neq 0$, taking $ a_{22}=-a_{33}$,  then \eqref{eq:3d-O-equ1} implies $a_{11}a_{21}-a_{21}a_{33}+a_{23}a_{31}+a_{21}a_{32}+a_{11}a_{31}+a_{33}a_{31} =0$. Then   \eqref{eq:3d-O-equ2}  and \eqref{eq:3d-O-equ3} imply $a_{21}+a_{31}+a_{33}^2+a_{23}a_{32}=0 $. Then there are two subcases: $(d_{221}): a_{31}=0$, and $(d_{222}): a_{31} \neq 0$.

$(d_{221})$ If $ a_{31}=0$, then \eqref{eq:3d-O-equ1} implies $a_{11}a_{21}+a_{21}a_{33}+a_{21}a_{32}=0$.  Then   \eqref{eq:3d-O-equ2}  and \eqref{eq:3d-O-equ3} imply  $a_{21}+a_{33}^2+a_{23}a_{32}=0 $.

Since $ a_{21}\neq 0$, taking $ a_{11}=0$,  then \eqref{eq:3d-O-equ1} implies $a_{32}=-a_{33}$.  Then   \eqref{eq:3d-O-equ2}  and \eqref{eq:3d-O-equ3} imply $a_{21}=a_{33}(a_{23}-a_{33}) $, we get $O_{29}$.

$(d_{222})$ If $ a_{31}\neq0$, $ a_{21}\neq0$, taking $a_{11}= a_{21}=a_{31}=1$,  $ a_{33}=0$,  then \eqref{eq:3d-O-equ1} implies $a_{32}=-a_{23}-2$. Then   \eqref{eq:3d-O-equ2}  and \eqref{eq:3d-O-equ3} imply $a_{23}a_{32}=-2 $.

If $ a_{23}=-1+\sqrt{3}$,  then  $ a_{32}=-1-\sqrt{3}$, we get  $O_{30}$.

If $ a_{23}=-1-\sqrt{3}$, then  $ a_{32}=-1+\sqrt{3}$, we get  $O_{31}$. This finishes the proof.

\end{document}